\documentclass[graybox]{svmult}

\usepackage{mathptmx}       
\usepackage{helvet}         
\usepackage{courier}        
\usepackage{type1cm}        
\usepackage{makeidx}         
\usepackage{graphicx}        
\usepackage{graphics}
\usepackage{multicol}        
\usepackage[bottom]{footmisc}

\makeindex             

\newcommand{\tn}[1]{{\color{black} {#1}}}
\newcommand{\Zed}{\mathbb Z}
\newcommand{\Real}{\mathbb R}
\newcommand{\Nat}{\mathbb N}
\newcommand{\Rat}{\mathbb Q}

\usepackage{amsmath}
\usepackage{amssymb}
\usepackage{mathtools}
\usepackage{epsfig, epstopdf}

\begin{document}

\title*{Broadcasting automata and patterns on $\Zed^2$}
\author{Thomas Nickson and Igor Potapov}
\institute{Thomas Nickson \at The University of Edinburgh, Kennedy Tower, Royal Edinburgh Hospital, Edinburgh EH10 5HF
 \email{tnickson@exseed.ed.ac.uk}
\and Igor Potapov \at University of Liverpool, Ashton Street, Ashton Building, Liverpool, L69 3BX, \email{potapov@liverpool.ac.uk}}
%
%
\maketitle

\abstract*{
The recently introduced Broadcasting Automata model draws inspiration from a variety of sources such as
Ad-Hoc radio networks, cellular automata, neighbourhood sequences and nature,
employing many of the same pattern forming methods that can be seen in the
superposition of waves and resonance.
Algorithms for the broadcasting automata model are in the same vain as those encountered in
distributed algorithms using a simple notion of waves, messages passed from automata
to automata throughout the topology, to construct computations. 
The waves generated by activating processes in a digital environment can be used
for designing a variety of wave algorithms.
In this chapter we aim to study the geometrical shapes of 
informational waves on integer grid generated in broadcasting automata model as well as their potential 
use for metric approximation in a discrete space.
An exploration of the ability to vary the broadcasting radius of each node leads
to results of categorisations of digital discs, their form, composition, encodings
and generation. Results pertaining to the nodal patterns generated by arbitrary
transmission radii on the plane are explored with a connection to broadcasting
sequences and approximation of discrete metrics of which results are given for the
approximation of astroids, a previously unachievable concave metric, through a
novel application of the aggregation of waves via a number of explored functions.
}

\abstract{
The Broadcasting Automata model draws inspiration from a variety of sources such as
Ad-Hoc radio networks, cellular automata, neighbourhood sequences and nature,
employing many of the same pattern forming methods that can be seen in the
superposition of waves and resonance.
Algorithms for broadcasting automata model are in the same vain as those encountered in
distributed algorithms using a simple notion of waves, messages passed from automata
to automata throughout the topology, to construct computations. 
The waves generated by activating processes in a digital environment can be used
for designing a variety of wave algorithms.
In this chapter we aim to study the geometrical shapes of 
informational waves on integer grid generated in broadcasting automata model as well as their potential 
use for metric approximation in a discrete space.
An exploration of the ability to vary the broadcasting radius of each node leads
to results of categorisations of digital discs, their form, composition, encodings
and generation. Results pertaining to the nodal patterns generated by arbitrary
transmission radii on the plane are explored with a connection to broadcasting
sequences and approximation of discrete metrics of which results are given for the
approximation of astroids, a previously unachievable concave metric, through a
novel application of the aggregation of waves via a number of explored functions.
}


\section{Introduction}

The recently introduced model of \emph{Broadcasting Automata}\index{broadcasting automata} \cite{Geometric} connects many of the techniques contained in distributed
algorithms, ad-hoc radio networks, cellular automata and neighbourhood sequences. Much like cellular
automata with variably defined neighbourhoods Broadcasting Automata can be
defined on some form of grid or lattice structure and have a simple computational
primitives comparative to a finite state automata with the ability to receive
and send messages both from and to those automata which are within its transmission
radius. Neighbourhoods are defined in the same way as with an ad-hoc
network, all those points within a certain transmission radius, receive the message
from the sender. 

Algorithms for the broadcasting automata model are in the same vain as those encountered in
distributed algorithms using a simple notion of waves, messages passed from automata
to automata throughout the topology, to construct computations \cite{Geometric2}. Wave
algorithms are enhanced further with notions of composition of the information
that is carried within each wave borrowed from the physical world and embellished
with the new found computational power of the automata.

The waves generated by activating processes in a digital environment can be used
for designing a variety of wave algorithms. In fact, even very simple finite functions for
the transformation and analysis of passing information provides more complex dynamics
than classical wave effects. In \cite{Geometric,Geometric2} we generalized the notion of the standing wave which is
a powerful tool for partitioning a cluster of robots on a non-oriented grid. In contrast
to classical waves where interference patterns generate nodal lines (i.e. lines
formed by points with constant values), an automata network can have more complex
patterns which are generated by periodic sequences of states in time.

Here we take a different direction and aim to study the geometrical shapes of 
informational waves on integer grid generated in broadcasting automata model as well as their potential 
use for metric approximation in a discrete space.
The repeated transmission to nodes, within certain radii, applied to
a certain network topology, or physical layout of nodes has been studied under the name, \emph{Neighbourhood Sequences}
(NS).
The concept of neighbourhood sequences is of importance in a number of practical applications and was originally applied for measuring distances in a digital world \cite{Hajdu20032597}.
Initially two classical digital motions (cityblock and chessboard)\footnote{ The cityblock motion allows movements only in horizontal and vertical directions, while the chessboard allows to move in diagonal directions.} were introduced. Based on these two types of motions
periodic neighbourhood sequences were defined in \cite{Das} by allowing arbitrary mixture of cityblock and chessboard motions.  In 2D the distances based on cityblock and chessboard neighbourhood sequences deviate quite substantially from the ideal Euclidian distances, so instead their combination that form ``the octagon'' was more often employed and studied.  Later the concept of neighbourhood sequences was extended to arbitrary finite and infinite dimensions, periodic and non-periodic sequences and then analysed in terms of their geometric properties.  

The aggregation of two classical neighbourhood sequences based on Moore and Von Neumann neighbourhoods (which correspond to cityblock and chessboard)  was recently proposed as an alternative method for self-organization, partitioning and pattern formation on the non-oriented grid environment in \cite{Geometric}. In particular the discrete analogs of physical standing wave phenomena were proposed to generate nodal patterns in the discrete environment by two neighbourhood sequences. The power of the primitives was illustrated by giving distributed algorithms for the problem of finding the centre of a digital disk of broadcasting automata. 
The shapes that can be formed by neighbourhood sequences in dimension two are quite limited. However the basic notion of neighbourhood sequences can be naturally extended by relaxing the constraints on the initial definition of the neighbourhood in such a way that two points are neighbors (r-neighbours) if the Euclidean distance is less than or equal to
some $r$, used to denote the radius of a circle. Then by \emph{Broadcasting sequences} we understand the periodic application of the $r$-neighbours distance function. 

The main result of this work is characterization of geometrical shapes that can be generated by Broadcasting Sequences on the square lattice. The shapes of r- neighborhoods correspond to Discrete Discs which are discrete convex polygons.
First we use the language of Chain codes (i.e. the code based on $8$ degrees of motion) to describe the shape of the Discrete Discs and their Broadcasting sequences. In particular we introduce the notion of Chain code segments and Line segments to express the shapes of the polygons corresponding to Discrete Discs. Then we characterize the shapes of polygons produced by Broadcasting sequences and provide linear time algorithm for the composition of two chain codes of Discrete Discs. Based on their composition properties we derive a number of limitations for produced polygons.  For example we show that there exist an infinite number of gradients (of line segments in the polygons) that cannot be produced by Broadcasting Sequences. It also becomes clear that the set of line segments, and as such gradients, that compose any discrete circle are closed under composition. 

Moreover, we provide an alternative method for enriching the set of geometrical shapes and neighbourhood sequences by aggregation of two Broadcasting Sequences. Initially we illustrate the idea on Moire and Anti-Moire aggregation function to produce an infinite family of polygons and polygonal shapes and characterize the gradients of their line segments. We have noticed that Anti-Moire aggregation function can be slightly modified to provide better approximation for the Euclidean distance on a square lattice then classical neighbourhood sequences. 
Finally it is possible to observe the variety of effects that are the result of the application of an aggregation function 
which are themselves shapes of some form.  

\section{Broadcasting Automata Model}

\tn{
One of the fundamental models of computation is the automaton. 
%
%
Finite state automata take as input a word, which may in some instances be referred to as a tape, 
and output is limited to an accepting state which the automata is left in if it is said to accept the word. Traditionally automata used in cases where it is important to transform some input in to an output, beyond the use of a single state, is the use of Moore Machines. Such machines differ from finite state machines in that they are able to produce an output word from an input word.

\begin{definition}
A Moore machine\index{Moore machine} \cite{conway2012regular} is a 6-tuple, $A = (Q,\Sigma,\Lambda,\delta,\Delta,q_0)\ $, where:
\begin{itemize}
\item $Q$ is a finite set of states, 
\item $\Sigma$ is the set of input symbols, 
\item $\Lambda$ the set of output symbols, 
\item $\delta:Q\times \Sigma \rightarrow Q$ is the transition function mapping a state $q\in Q$ and a symbol, or set of symbols, $\sigma \in \Sigma$ to a state $q\in Q$, 
\item $\Delta:Q\rightarrow \Lambda$ is the output function which maps a state, $q\in Q$, to an output symbol, $\lambda \in \Lambda$, and 
\item $q_0$ is the initial or quiescent state in which the automata starts.
\end{itemize}
\end{definition}

It is assumed that such a machine is connected to an input tape, or word, and an output tape, or word, where the result of the computation is read. In a situation, as is presented in distributed systems, whereby automata are connected to each other it is possible for the output of one automata to become the input of another automata. Such a model is known as a network of automata and connections from one automatons output to another's input may be represented as a directed graph, where direction represents the output going to input from automaton to automaton.

\begin{definition}
A network of finite automata\index{network of automata} is a triple, $(G, A, C_0)\ $, where:
\begin{itemize}
\item $G=(V,E)$ is a directed graph, with vertices, $V$, and edges, $E$, which are ordered pairs of vertices,
\item $A$, is a Moore machine, and 
\item $C_0$ is an initial configuration which maps states, $Q$, of the automata, $A$, to vertices, $V$, such that, $C_0 : V \rightarrow Q$.
\end{itemize}
\end{definition}

In this model the topology is fixed as specified by the construction of the graph, $G$, which dictates the flow of inputs and outputs from the Moore machines. Such symbols, where a symbol is part of the input/output word of the Moore machine, are generated as response to some input, by an automaton at vertex, $v\in V$, and then sent to all of the adjacent vertices in the graph or to a particular adjacent vertex, $G$, where the automata that receive the symbols process them as they would their input word. 

A less abstract model is considered by adding more details about the communication between automata, 
where they are located in Euclidean space and what can be transmitted. 
It will also be required to define more refined models that will reflect the new features and
constraints of the physical environment, but are still at a high level of abstraction.

Taking the inspiration from ad hoc wireless communication networks we introduced in \cite{Geometric} 
a new model of \emph{Broadcasting Automata}\index{broadcasting automata}, 
which can be seen as a network of finite automata with a dynamic network topology. 
Informally speaking, the model of Broadcasting Automata comprises nodes, which correspond to points in space, and connectivity between nodes depends upon the distances between the points and the
strengths of the transmissions generated by the automata as may be seen in ad-hoc
radio networks.
Transmission strength may vary from round to round for any particular automaton in the space and is dictated by the state of the automaton. In this model the topology, or connectivity graph, of the network of automata is able to change at each time step based on the states of the automata. 

In order to give a formal definition of the Broadcasting Automata model on a metric space\index{metric space} it is first necessary to introduce the notion of a metric space and to modify the classical notion of the Moore machine.

\begin{definition}
A metric space is an ordered pair, $(M,d)$, where $M$ is a set and $d$ is a metric on $M$ such that $d: M\times  M\rightarrow \Real$.
\end{definition}

Here, $(M,d)$, is any metric space, later the two dimensional euclidean space shall be used, 
but this is not a necessity simply that a notion of the distance between two points is required.

\begin{definition}
The \emph{Broadcasting Automaton} extends the Moore Machine by introducing a set of final states, $F$, along with a function, $\tau : Q \rightarrow \Real$, which maps the state to a real number and represents the radius of transmission for the output symbol and the output alphabet, $\Lambda$, is extended by adding an empty symbol, $\epsilon$ and it is represented by an 8-tuple $A = (Q,\Sigma,\Lambda,\delta,\Delta,\tau,q_0,F) $.
\end{definition}

A network of \emph{Broadcasting Automata}, which will also be referred to as the \emph{Broadcasting Automata model}, can now be defined. 

\begin{definition}
The \emph{Broadcasting Automata model} is represented by a triple $BA = ((M,d),A,C_0)$ where:
\begin{itemize}
\item $(M,d)$ is a metric space,
\item $A$ is a Broadcasting Automata, $A = (Q,\Sigma,\Lambda,\delta,\Delta,\tau,q_0,F)$,
\item $C_0$ is the initial configuration of the Broadcasting Automata model, it is a mapping from points, $M$, to states of the Broadcasting Automata, $Q$, such that, $C_0 : M \rightarrow Q$.
\end{itemize}
\end{definition}

In some cases it may be that the input and output symbols are drawn from the same alphabet in which case, $\Sigma = \Lambda$.

The communication between automata is organised by message passing, where \emph{messages} are symbols from the output alphabet, $\Lambda$, of the automata, $A$, to all of the automata within its transmission radius. Messages, \tn{symbols from the output alphabet, $\Lambda$, of the automata, $A$,} are generated and passed instantaneously at discrete time steps, \tn{generation of message is given by the function, $\Delta$, for the automata, $A$, resulting in synchronous steps}. Those automata that have received a message, \tn{for the first time in the computation}, are said to be \emph{activated}. 

If several messages are transmitted to an automaton, $A$, it will receive only a set of {\bf unique messages}, i.e. for any multiset of transmitting messages, where the multiset represents a number of the same message being sent, received by $A$, over some number of rounds, the information about quantity of each type, within a single round, will be lost. This simply illustrates that the automaton may not dictate a number
of the same symbol in any transition, all transitions must operate upon a set of
distinct symbols.

More formally the concept of computation and communication is captured in the concept of configurations of the Broadcasting Automata model and its semantics presented here.

\begin{definition}\label{config}
The \emph{configuration} of the Broadcasting Automata is given by the mapping, $c:M\rightarrow Q$, from points in the metric space to states. It can be noted here that $C_0$ is an initial configuration for the Broadcasting Automata model.
\end{definition}
}
\tn{Automata are updated, from configuration to configuration, synchronously at discrete time steps. The next state of each automaton depends upon the states of all other automata which have in their neighbourhood the automaton that is to be updated. Where here the neighbourhoods are formed by a combination of their position in $M$ and the range of their transmission, dictated by $\tau(q)$ for all automata. 

\begin{definition}\label{messageRec}
The set of messages received by the automaton at a point, $u,v \in M$, is expressed by the set $\Gamma_u= \{ \Delta(c(v)) | v\in M \wedge d(u,v)\leq \tau(c(v))\}$ for the metric space, $(M,d)$.
\end{definition}

\begin{figure}[htp]
\centering
\includegraphics[scale=0.33]{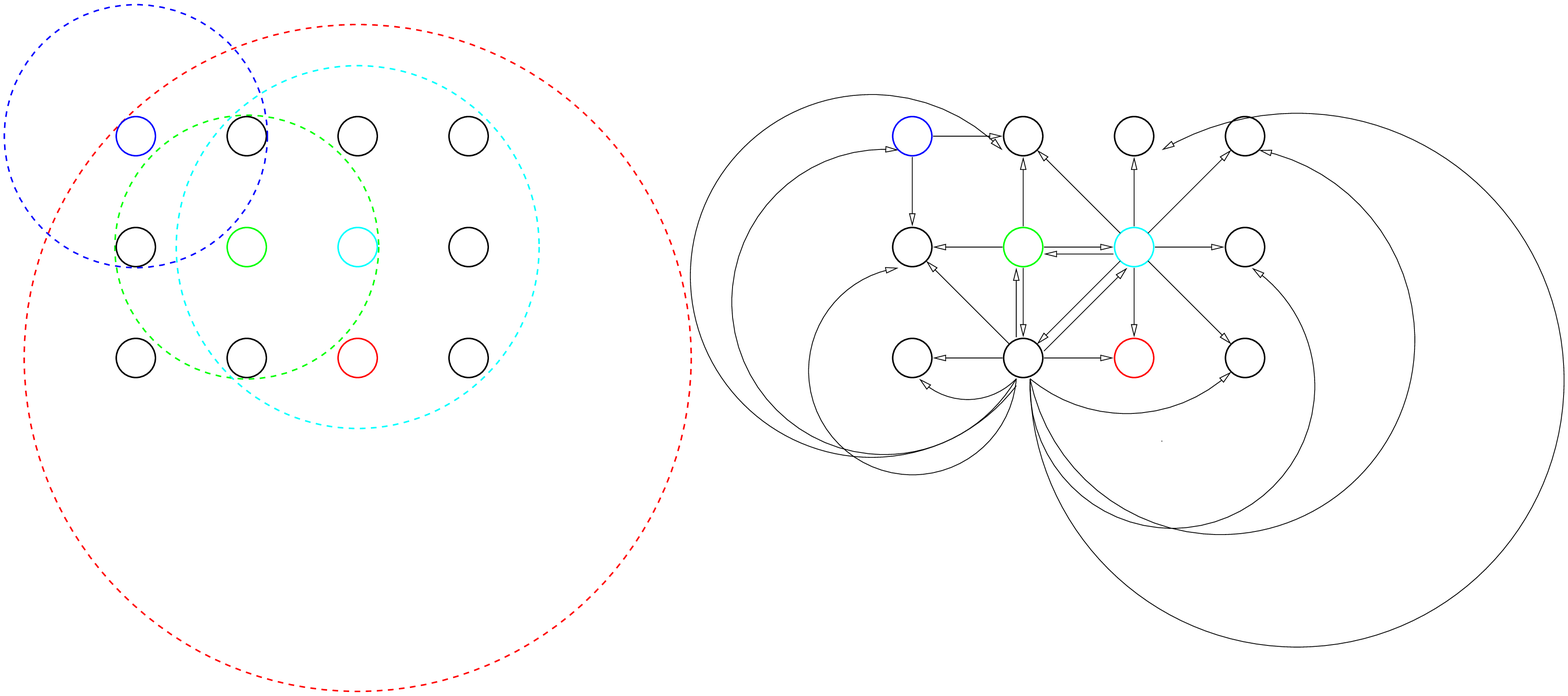}
\includegraphics[scale=0.33]{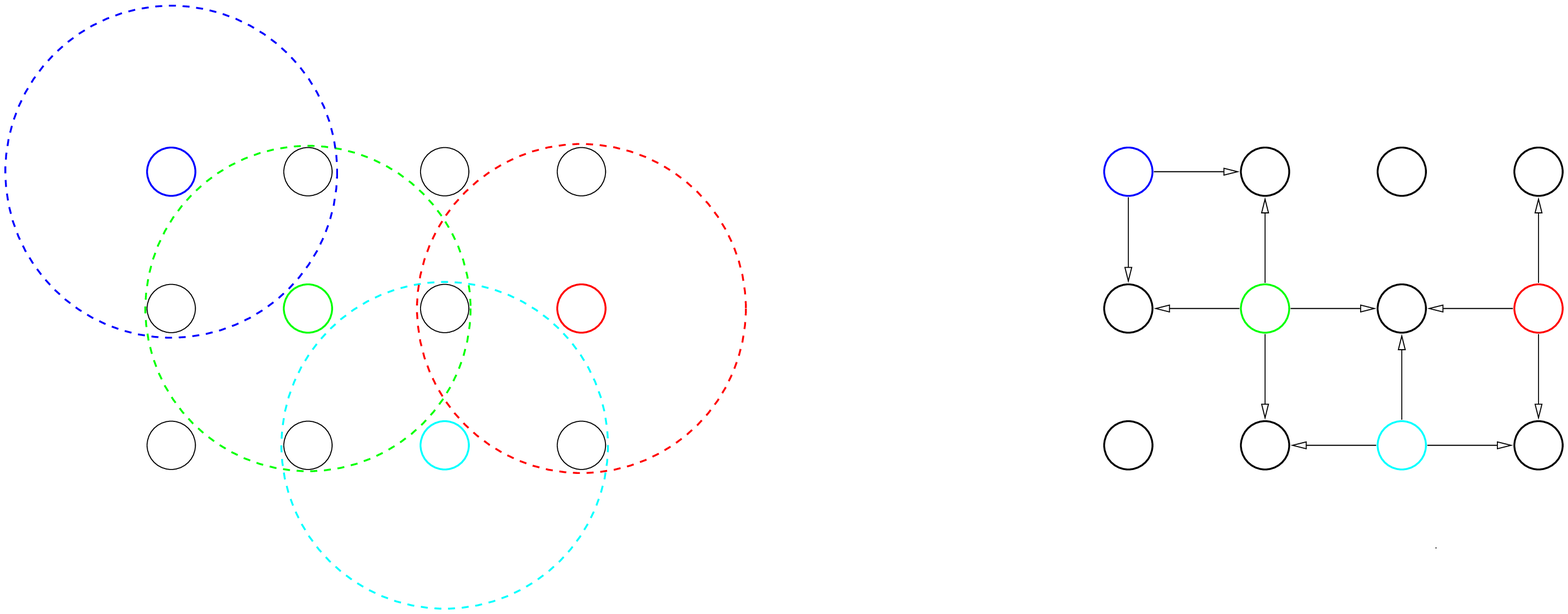}
\caption{The above figures show the possible evolution over time of a network of Broadcasting Automata\index{network of automata}. Each figure shows, on the left, the broadcast range for the automata, depicted by a dotted circle with the node at the centre, on the right, the same connectivity is shown in graph form. \label{transGraph}}
\end{figure}

In Figure~\ref{transGraph}, it is possible to see an elucidation of message passing in the Broadcasting Automata model through the process of constructing a digraph representation of the broadcasting radii. The figure depicts the automatons broadcasting radius by way of a dashed circle with the automaton broadcasting at that range shown in the centre. The corresponding graph may be constructed from this by making a directed edge in the graph from the node that is broadcasting to all nodes that are within the broadcast range, where the range is dictated by the state, $q\in Q$ and the function, $\tau(q)=r$, where $r$ is the broadcasting radius. The construction of such a graph shows both the connection to the network automata model, which operates on a similar premise with a fixed graph, but also highlights how Definition~\ref{messageRec} is constructed. Incoming edges are generated by connecting nodes, with the arrow from transmitter to node, reachable from the transmitter. Naturally such a graph, in the Broadcasting Automata model, can change over time which is shown by images further down the page as increase time steps in Figure~\ref{transGraph}.

It should be noted that this set of messages can be empty. The new state of the automaton at $u$ is now given by applying the automaton's transition function $\delta_{u}(\Gamma_{u})$. This may be applied to all of the automata in the BA and as such determines the global dynamics of the system. It can be seen that in one time step it is possible for configuration $c$ to become configuration $e$ where, for all $v\in M$, $e(v)=\delta_{v}(\Gamma_u)$. Such a function is called the global transition function. The { set of all configurations for BA} is denoted as $\mathcal{C}$.

\begin{definition}\label{transition}
The global transition function, $\mathcal{G}:\mathcal{C}\rightarrow \mathcal{C}$, represents the transition of the system at discrete time steps such that for two configurations, $c\in \mathcal{C}$ and $e\in \mathcal{C}$, where $e$ is the configuration the time step directly after $c$ such that, $e = \mathcal{G}(c)$.
\end{definition}

\begin{definition}
A computation is defined as any series of applications of the function, $\mathcal{G}$, from an initial condition, $C_0$, that leads to the set of all automata in the system being in one of their accepting states from the set, $F$, or the computation halts such that there exists no defined relation from the current configuration, $c$, under the global transition function $G(c)$. 
\end{definition}

Here a notion of reachability is given.

\begin{definition}
One configuration, $c$, is said to be {\bf reachable} from another, $e$, if from the initial configuration, $C_0$, there exists a series of valid intermediary configurations such that $c\rightarrow c'\rightarrow c''...\rightarrow e$. Where $c\rightarrow c'$ is equivalent to $\mathcal{G}(c) = c'$. It can also be stated, in a shorthand way, that where the sequence $c\rightarrow c'\rightarrow c''...\rightarrow e$ exists reachability from $c$ to $e$ may be shown as $c\leadsto e$.
\end{definition}

In the above context the \emph{time complexity} of an algorithm in the Broadcasting Automata model is determined by the number of applications of the global transition function during the  computation.

%
}


Whilst in general Broadcasting Automata may be used represent any of the common network topologies (such as, linear array, ring, star, tree, near-neighbour mesh, systolic array, completely connected, chordal ring, 3-cubes and hyper-cubes \cite{44900, 1667197}), through varying the location and transmission radii of the automata in the space, a different model is employed. 
The model used throughout the paper is a hybrid mixing the notion of transmission radius in the pathloss geometric random graph model and common network topologies. The method involves the physical placement of nodes on the Euclidean plane in a lattice structure such that the distances between points conform to the euclidean distance. Such regular lattices coincide with those structures that can be formed by previous works \cite{Spring08, Suzuki99distributedanonymous, JoRaM09} giving the basis of the forms of lattice that can be studied. This paper will strictly concern itself with the {\bf square grid lattice}\index{grid lattice} though it can be extended in to any of the other grid configurations such as the hexagonal or triangular lattices along with many others. 

The locations of the points on the lattice allow the construction of the metric space, $(M,d)$ as such the locations of the automata. Assuming a lattice structure for positioning of the automata and transmission radii, $\mathcal{R}$, is in accordance with the principles of transmission used in the pathloss geometric random graph model put forth in Ad-Hoc networks. 

A variety of topologies, the lattices that can be constructed as in \cite{Spring08} (square, hexagon and triangle), and the shapes of neighbourhoods generated by varying the transmission neighbourhood can be seen in Figure~\ref{neighLatt} where black dots represent the placement of automata on the plane according to the lattice here shown by the black lines. Neighbours are those within the dashed circle and the automata at the centre of the circle is the initial transmitter of any messages.

\begin{figure}[htp]
\centering
\includegraphics[scale=0.55]{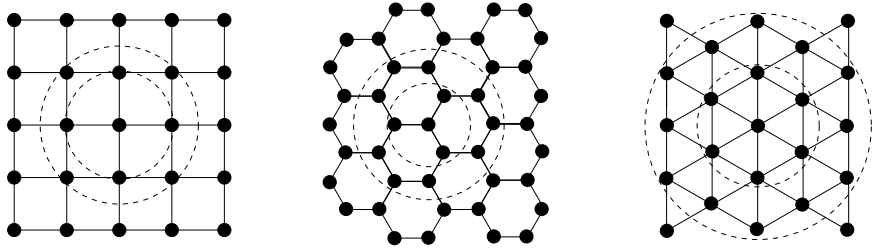}
\caption{Showing three differing lattices and how their neighbourhoods can be altered simply by varying the radii of the transmission neighbourhood which, here, is depicted as a dotted circle. All automata, shown as black dots, are able to receive messages from the automata at the centre of the circle.The three lattices depicted are (left to right) square, hexagon and triangle.\label{neighLatt}}
\end{figure}


The connectivity graph defined in ad-hoc radio networks denotes those that are able to send and receive messages with certainty however in distributed algorithms the possibility of communication errors is not ignored. In such models messages may be lost, duplicated, reordered or garbled where they must be detected and corrected by supplementary mechanisms which are mostly referred to as protocols for example the common network protocol of TCP. The Broadcasting Automata model sides with the former method of message passing. It is assumed that there will be no error in transmission, the receipt or sending of messages across the network is guaranteed however consideration is made to the synchronicity of the network due to the high sensitivity to the timing of message passing as will become clear later when discussing pattern formation with such protocols.

To this extent we consider two variants of Broadcasting Automata: {\bf synchronous} and {\bf asynchronous} (or reactive) models. In each of the variations both the transmission of a method and amount of time take for the automata to process the message is considered to be constant. It is this constant time that will afford synchronisation within the model. \tn{The two differing models are presented to show that there are possible complexity trade offs depending on how it is to be measured. If the size of the alphabet is a priority, such that it must be minimised, then the synchronous model is best, however this comes at a penalty of time complexity, where the asynchronous model fairs better.}

In the {\bf asynchronous} model upon the receipt of a message, an input symbol (or set of symbols) $\sigma\in \Sigma$, the automaton, $A$, becomes active. The automaton may only react to a non-empty set of messages, it may not make and epsilon transition. Once activated the automaton, $A$, reacts to the symbol(s), $\sigma$, according to the configuration of the automaton and responds with the transmission of an output symbol, $\lambda\in \Lambda$, to \tn{those in its transmission radius for the current state, $q\in Q$, $\tau(q)$}. The processing of the input symbols by the automaton is done in one discrete time step which is the same period for all automata in the graph. Once active the automaton must wait for another message in order to change its state it is not allowed epsilon transitions.

The {\bf synchronous} model has a singular alphabet and as such, $\Sigma = \Lambda$ with the allowance of epsilon transitions. Upon the receipt of a message, an input symbol (or set of symbols) $\sigma\in \Sigma$, the automaton, $A$, becomes active. Once active the automaton is synchronised with the rest of the graph by repeated epsilon transitions that are representative of a constant transition via input symbol as in the asynchronous model. As each such transition is considered a change in configuration of the automaton it must take a single time step. As such this guarantees the synchronisation of system. However the multiple alphabet can be simulated by associating different symbols to different time steps as shall be seen later.


The following outlines, informally, the methodology of message passing used by the two differing constructions of automata, synchronous and asynchronous.

\noindent In the {\bf synchronous} model messages are passed from automaton to automaton according to the following rules: 
\begin{enumerate}
\item An automaton, $A$, receives a message from an activating source at time, $t$; 
\item At time $t+1$, $A$ sends a message to all automata within its transmission radius dependent on its state, $q\in Q$, $\tau(q)$;
\item At time $t+2$, $A$ ignores all incoming messages for this round.
\end{enumerate}
In the {\bf asynchronous model} the following rules can be applied:
\begin{enumerate}
\item An automaton, $A$, receives a message $\sigma_i\in \Sigma$ from an activating source at time $t$;
\item The automaton, $A$, broadcasts a message $\sigma_{(i+1)\mod |\Sigma|}$ to all automata within its transmission radius, dependent upon its state, $q\in Q$, $\tau(q)$ at time $t+1$; 
\item Ignore all incoming messages at time $t+2$. 
\end{enumerate}

\tn{In both models Step 3 prevents an automaton from receiving back the message that has just been passed to the automatons neighbours by ignoring all transmissions received the round after transmission and ensuring that the messages are always carried away from the initial source of transmission. In both cases this rejection of all messages may be modelled by the addition of a state that simply does not accept input in that state, where epsilon transitions apply, or that the transition is made to another equivalent state independent of input received, where from here it is possible to receive transmissions that once again affect the logic of the program. }


Naturally, given these two models it is interesting to see if they are capable of the same computations and that indeed anything that can be done in one model can be done in the other. There are many ways with which to establish equivalence between automata and models in general. One important technique that has been used to show equivalence in connection with Turing machines is {simulation} \cite{linz2001introduction}. Here a similar proposition is made with respects to establishing the equivalence of the two models of automata that are given here, synchronous and asynchronous. 
%

\begin{proposition}\label{equiv}~\cite{Geometric2}
Both the synchronous and asynchronous models are able to simulate the other.
\end{proposition}



It should also be noted that with some generalisation, to allow a variable radius neighbourhood interaction which has been explored in a limited sense \cite{Wolfram19841, Gerhardt1990392}, the Cellular Automata\index{cellular automata} (CA) model may also be used to simulate the Broadcasting Automata model. In a grid of CA it is possible to assign states that correlate to transmitters where each transmitter in the grid is assigned a unique state and all other CA are in the quiescent state. The initial transmitters state causes each of the automata who have the initial transmitter in their neighbourhood to change from the quiescent state to the transmitters state plus one over some modulo. This clearly simulates the transmission of some word over a modulo for all automata in the grid. Upon finding a CA with one of this initial set of states within its neighbourhood, and such that it is all encompassing in its neighbourhood, the second transmitter changes its state which triggers a second cascade of state changes that are equivalent to those from the first transmitter only this time the automata must take in to account the initial state that it is already in. This should provide an exact copy of the Broadcasting Automata model, assuming that CA is extended to allow such variable neighbourhoods. This does not mean however that there is an exact translation of Cellular Automata in to Broadcasting Automata and as such it is only possible to say that $BA \subseteq CA$. 
%
%
%
%
\section{Variable Radius Broadcasting over $\Zed^2$}
%
%
Whilst the model provided here covers many configurations for the underlying communication graph this paper is mainly considering the following case, where all automata are place in euclidean space according to a regular arrangement equivalent to that of a square lattice such that, using the Cartesian coordinate system, for any automata in the space its nearest neighbour is at distance $1$ and its next nearest neighbour is at distance $\sqrt{2}$. This idea is illustrated for two dimensions in Figure~\ref{radii}. \tn{This now leads to the model of Broadcasting Automata where $(M,d)$ is $M = \Zed \times \Zed$ and where for $\rho=(\rho_0,rho_1)$ and $\rho'=(\rho'_0,\rho'_1)$ the distance function is $d(\rho,\rho')=\sqrt{(\rho_0-\rho'_0)^2+(\rho_0-\rho'_0)^2}$, the Euclidean distance function for $\rho,\rho'\in M$.}
\begin{figure}[htp]
\centering
\includegraphics[scale=0.22]{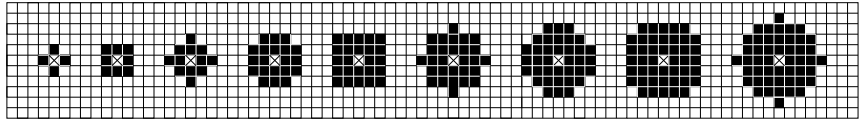}
\caption{A variety of transmission radii are shown (l-r) squared radii $r^2=\{1,2,4,5,8,9,10,13,16\}$. Crosses represent the centre of the respective discrete disc.\label{radii} }
\end{figure}

\tn{
Figure~\ref{propWaves} shows as example the automata that receive a message when the Euclidean space and square lattice are restricted to two dimensions but the radius of broadcast, $\tau(q)$ for $q\in Q$, is varied. The automaton at point $\rho\in M$ which is the source of the transmission is shown as the circle at the centre of the surrounding automata on the plane, those that are black are within the transmission range $\tau(q)$ for $q\in Q$.  Successive larger collections of automata coloured black show what happens when range can be changed to alter the automata that are included in the transmission radius, $\tau(q)\in \mathcal{R}$. If the transmission radius, $\mathcal{R}$, is equal to 1, as in Figure~\ref{propWaves} diagram $a)$ then only four of the eight automata can be reached. If the radius is made slightly larger and is equal to $\sqrt{2}$, it can encompass all eight automata in its neighbourhood as shown in diagram $b)$. Such structures are identical to the well studied neighbourhoods von Neumann\index{von Neumann neighbourhood} and Moore\index{Moore neighbourhood} respectively and as such here it shall be considered that such constructions of neighbourhoods are a generalisation of these two neighbourhoods. As we will show later, iterative broadcasting within von Neumann and Moore neighbourhoods can distribute messages in the form of a diamond wave and a square wave as shown later in Figure~\ref{fig_MVN}.}

\begin{figure}[htp]
\centering
\includegraphics[scale=0.3]{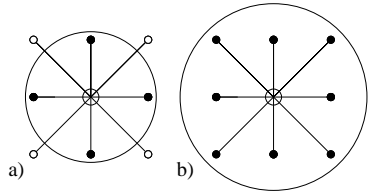}
\caption{ Diagram $a)$ represents the propagation pattern for a diamond wave (Von Neumann neighbourhood) and 
diagram $b)$ shows the propagation pattern for a square wave (Moore neighbourhood). 
\label{propWaves}}
\end{figure}


The construction of distinct radii for the circles is defined by the numbers $n$ such that $n = x^2 + y^2$ has a solution in non-negative integers $x$, $y$ \cite{oeisA001481}. Here it can be seen that $r^2$ is given for convenience as $n$ always has a convenient representation in $\Zed$ whereas it becomes cumbersome to write out either the root of the integer or its decimal representation. For an explanation as to how these numbers relate to the distinct discrete discs it must be noted that distinct discs are constructed such that $r^2$ contains a new solution in $x$ and $y$. As it can be considered that the digital disc represents all of the maximal combinations of Pythagorean triples such that $x^2+y^2\leq r^2$ then for one disc to differ from another there must be an increase of $r^2$ such that there exists a new, distinct, maximal solution for $x,y\in \Zed$. Generating these in the inverse direction by choosing $x,y\in \Zed$ such that it is a new maximal combination not contained in any of the previous, smaller discrete discs allows the construction of the list of $r^2$ that defines the sequence of distinct discrete discs. 

%


\section{Neighbourhood and Broadcasting Sequences}

The idea of propagating a pattern of square ($r^2=2$) or diamond waves ($r^2=1$), generated by repeated application of Moore or Von Neumann waves respectively, in dimension two are commonly known as Neighbourhood Sequences (see Figure~\ref{fig_MVN}). Neighbourhood Sequences\index{neighbourhood sequences} are an abstraction used to study certain discrete distance metrics that are generated upon the repeated application of certain neighbourhoods to a lattice for example the Moore neighbourhood to a square lattice. As previously discussed such neighbourhoods as von Neumann and Moore are encompassed by the more general model used here whereby the transmission radius is varied on a square lattice the first two of such transmission radii to produce distinct objects equating to those of the von Neumann and Moore neighbourhoods.

\begin{figure}[htp]
\centering
\includegraphics[scale=0.3]{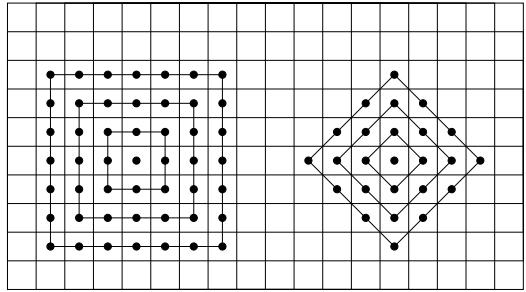}
\caption{Wave propagation with Neighbourhood Sequences on the square grid: Moore neighbourhood (left) and Von Neumann neighbourhood (right).
\label{fig_MVN}}
\end{figure}

Definitions and notation concerning neighbourhood sequences as considered in \cite{FarBajNag06}, and many other works, are now 
given here for completeness.

Let $p\in \Zed^n$ where $n\in \mathbb{N}$ and such that the $i$th coordinate of $p$ is given by $Pr_i(p)$ for $1\leq i \leq n$. 

\begin{definition}\label{mNeigh}
For $M\in \Zed$ where $0\leq M \leq n$ the points $p,q\in \Zed^n$ are $M-Neighbours$ when the following two conditions hold:

\begin{itemize}
\item $|Pr_i(p)-Pr_i(q)|\leq 1$ for $(1\leq i \leq n)$
\item $\sum_{i=1}^{n}|Pr_i(p)-Pr_i(q)|\leq M$
\end{itemize}
\end{definition}

\tn{An $n$-dimensional neighbourhood sequence is denoted $\mathcal{A}=(a(i))^{\infty}_{i=1}$, $\forall i\in \mathbb{N}$ where $a(i)\in {1,...,n}$ denotes an \emph{M-neighbourhood} by its value of $M$ such that for $a(1)$ denotes that the next neighbourhood set of points in the sequence is those that differ by at most one coordinate as given by Definition~\ref{mNeigh} where $M=1$ in this case the Von Neumann neighbourhood. If $\mathcal{A}$ is periodic then $\exists l\in \mathbb{N}$, $a(i+l)=a(i)$ $(i\in \mathbb{N})$. Such periodic sequences are given as $\mathcal{A} = (a(1),a(2),...,a(l))$. 

\begin{definition}
The $\mathcal{A}$-distance, $d(p,q;\mathcal{A})$, of $p$ and $q$ is the length of the shortest A-path(s) between them. 
\end{definition}

As the spreading of such neighbourhoods is translation invariant only an initial point of the origin need be considered w.l.o.g. 

\begin{definition}
The region occupied after $k$ applications of the neighbourhood sequence $\mathcal{A}$ is denoted as $\mathcal{A}_k=\{p\in \Zed^n:d(0,p;\mathcal{A})\leq k\}$ for $k\in \mathbb{N}$.
\end{definition}

Also, let $H(\mathcal{A}_k)$ be the convex hull, given in Definition~\ref{convHull}, of $\mathcal{A}_k$ in $\Zed^n$.

\begin{definition}\label{convHull}
 A convex hull is the smallest set of points that form a convex set as given in Definition~\ref{convSet}.
\end{definition}

\begin{definition}\label{convSet}
{\bf Convex set}. A set $C\subset  R^d$ is convex if for every two points $x, y \in C$ the whole segment $xy$ is also contained in $C$. In other words, for every $t \in [0, 1]$, the point $tx + ( 1 - t ) y$ belongs to $C$.
\cite{mat2002lectures}
\end{definition}

}

Discrete discs are formed by the Broadcasting Automata as they are arranged on the plane at integral Cartesian coordinates, $(\Zed \times \Zed)\in M$, and as such a broadcast to all automata in range, for point $v\in M$, $\tau(c(v))=r$ will cause a change in state to all those automata that form a discrete disc of radius $r$ from the point $v$. In the broadcasting automata model it is assumed that there is no restriction on the radius of transmission where it is considered that as the Von Neumann and Moore neighbourhoods can be described as the first two radii in the set of distinct discrete discs, $r^2=1$ and $r^2=2$. This means that discrete discs are quite a natural extension to the basic notion of neighbourhood sequences merely relaxing the constraints on the initial definition of $M-neighbours$ in the following way. 

\begin{definition}\label{rNeighbours}
Two points $p,q\in \Zed^n$ are $r-neighbours$ if the Euclidean distance, $d(p,q)=\sqrt{\sum^{n}_{i=0}(q_i-p_i)^2}$, is less than some $r$, used to denote the radius of a circle, such that $d(p,q)\leq r$. 
\end{definition}

Utilising the framework supplied by the work done on neighbourhood sequences it is no possible to outline {\bf Broadcasting sequences}\index{broadcasting sequences} which denote any sequence of radii, $r$, such that $R=(r_1,r_2,...,r_l)$. Labelling those points that are reachable by some application, $R_k$, of the neighbourhood sequence is another extension to the notation. All points such that $p\in R_1$ are labelled $0$, all points $p\in R_2\backslash R_1$ are labelled as $1$. More generally labels will be assigned $b$ where $k \equiv b \mod m$ and $m\in \mathbb{N}$. The work conducted here will also be restricted to the study of $\Zed^2$.

Many of the questions asked and, indeed, answered, in neighbourhood sequences, shall be useful in the exposition of broadcasting automata. These include whether a certain sequence of neighbourhoods are metrical or not \cite{1521291}. Indeed when using a general form of broadcasting automata whereby alternation of the broadcasting radius is allowed at each step as suggested here. Considering that all notions of the construction of algorithms presented here in the broadcasting automata model rely on the distance from the transmitter to the node, being able to assure that this distance will be consistent for all automata is essential. Also studied in neighbourhood sequences is the possible resultant shapes which are categorised and examined for their various properties such as their use in approximation\index{metric approximation} of euclidean distances where the isoperimetric ratio is used to compare the sequences based on the shapes that they form on the plane \cite{2010arXiv1006.3404F}. This same method is used to again show and estimation for euclidean distance as well as estimations for certain $L_p$ distances, here the astroid which will be presented in Section~\ref{Lp_Metrics}. The characterisation of the shapes generated by these neighbourhood sequences yield results about the shapes of compositions of such discrete discs on the plane and the partitions that such compositions form.


\section{Geometrical Properties of Discrete Circles}

The characterization of broadcasting sequences on ${\Zed}^2$ is closely related to the study of discrete circles on the square lattice
and their discrete representation. One of the efficient methods to describe the discrete circle is a chain coding.
%
%
%
%
The method of chain coding was first described in \cite{5219197} as a way in which to encode arbitrary geometric configurations where they were initially used, as they shall be here, to facilitate their analysis and manipulation through computational means. At the time there were many questions about the encoding of shapes and the methodologies which should be used to encode such shapes and chain coding presented an schema which was simple, highly standardised and universally applicable to all continuous curves. Whilst the definition given here differs from the definition given in \cite{5219197} it only differs for convenience when discussing discrete discs in the first octant allowing the use of only the numeric symbols $0$ and $1$ instead of $0$ and $7$ as Freeman suggested. This gives the following definition of chain coding that shall be used throughout.

\begin{definition}
From a starting point on the square lattice a {\bf chain code}\index{chain code} is a word from the alphabet $\{0,1,2,3,4,5,6,7\}$.
\end{definition}

\begin{figure}[htp]
\centering
\includegraphics[scale=0.45]{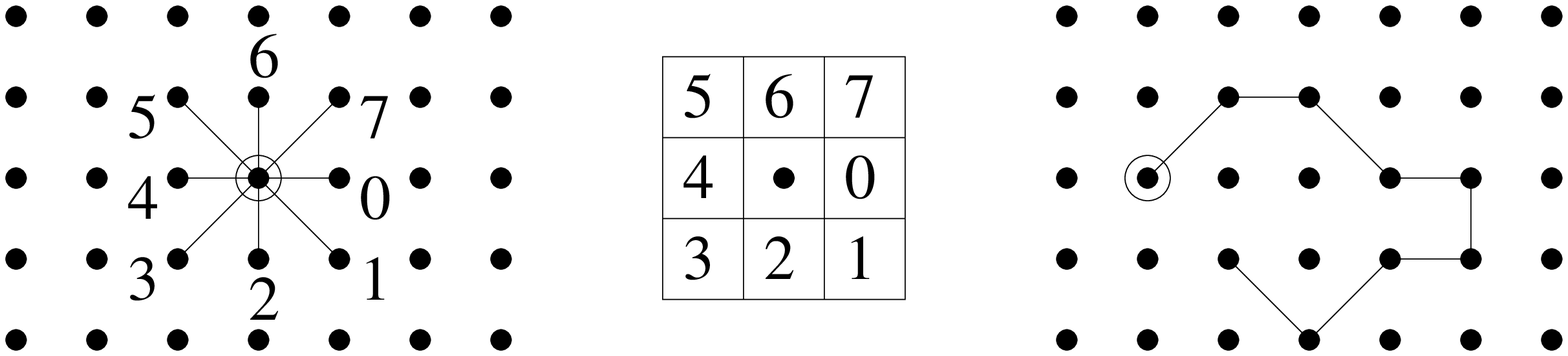}
\caption{(Left) Showing the eight possibilities of motion from the single point of origin, circled, as they may be traced out on the plane (Centre) Shows the eight possible points of motion from a central point and (Right) gives and example of an arbitrary line traced out on the lattice with a chain code of form $70102435$ originating from the circled point. \label{chainCod}}
\end{figure}

From some starting point chain codes may be used to reconstruct a shape, $S$, by translating the code's motion on to the lattice where $0$ indicates positive movement along the $x-axis$ with increasing values moving clockwise through all $8$ possible degrees of motion on the lattice, labelled from $0$ to $7$ respectively. An integral shape, $S$, can be encoded using an inverse method whereby the shape is traced from some starting point and the motion from one point to the next connected point in the shape is recorded as a chain code. For example, a straight line would be represented by the infinite repetition of a single digit such as $00000000000000...$, a line that satisfies the equation, $x-y=0$, may be encoded $77777777...$ or as the example given in Figure~\ref{chainCod} and arbitrary line may be encoded numerically as $7010235$.


The discrete disc is the basic descriptor of a transmission in broadcasting automata. It represents the total set of all automata that could be reached on the square lattice after a single broadcast of radius, $r^2$. In the following sections reasoning will follow only from the first octant, which can be seen in Figure~\ref{octants}, of the  discrete circle which shall now be defined and gives all the information that is required to categorise the discrete disc, its perimeter, but also has a convenient method of chain code construction. The need only to observe the first octant is derived from the symmetry that occurs such that it is possible to take a solution of the first octant and by successive reflection operations about the eight octants generate a full circle \cite{Bresenham:1977:LAI:359423.359432} under the assumption, which is also made here, that the circles are centred on a integral point on the plane.

\begin{definition}\label{circD}
A discrete circle is composed of points in the $\Zed^2$ set $\zeta =\{(x,y) | x^2+y^2\leq r^2,\ x,y\in \Zed,\ r\in \Real \}$ that have in their Moore neighbourhood any point in the complement of $\zeta$.
\end{definition}

\begin{figure}[htp]
\centering
\includegraphics[scale=0.60]{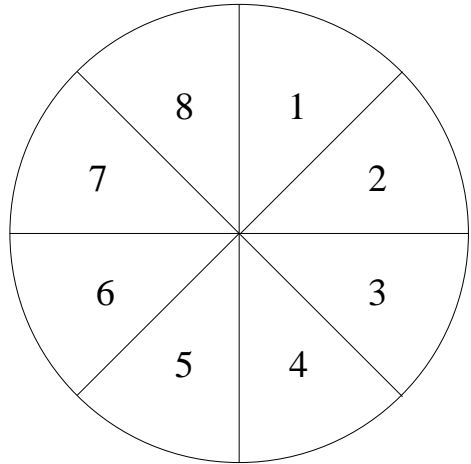}
\caption{An illustration of a circle that has been divided in to eight octants with the labelling of those octants applied accordingly.\label{octants}}
\end{figure}

\begin{figure}[htp]
\centering
\includegraphics[scale=0.80]{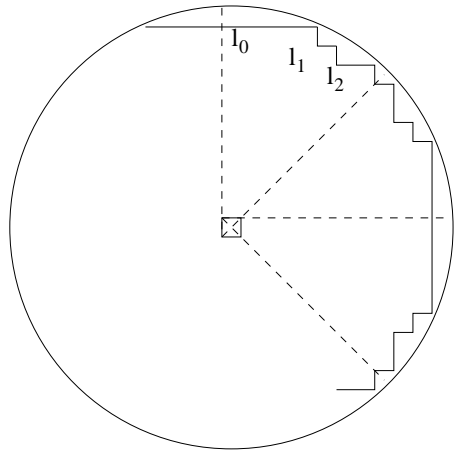}
\caption{An illustration of a discrete disc and its accompanying discrete circle of radius squared 116, where the chain code for the first octant is labelled as $l_0l_1l_2$.}\label{digiCirc}
\end{figure}

In Definition~\ref{circD} it is shown that $r\in \Real$ however many of these do not produce distinct discrete circles. The set of $r$ that includes all distinct circles are those which form Pythagorean triples which consist of two positive integers $a,b\in \Nat$, such that $a^2+b^2=r^2$.

A simple way to draw the discrete circle is to divide it into octants as can be seen in Figure~\ref{digiCirc}. First calculate only one octant of the circle and as the rest of the circle is "mirrored" from the first octant it can now be successively mirrored to allow the construction of the other sections of the circle. Let us divide the circle into eight octants labelled 1-8 clockwise from the y-axis. The notation $Octant(i)$ is used to refer to octant sector $i$. If an algorithm can generate the discrete circle segment on one octant, it can be used to generate other octants by using some simple transformations, such as rotation by $90^{\circ}$ and reflection around the $x$-axis, $y$-axis, and the $+45^{\circ}$ diagonal lines. 

Some basic results are now given about chain codes of such discrete circles. It is possible to show that in the first octant and given the definition of chain coding presented here only two digits, $0$ and $1$ are required to fully represent the discrete circle.

\begin{lemma}\label{chainStruc}
A Chain coding $w$ for the circle in the first octant is a word in $\{0,1\}^*$. 
\end{lemma}

\begin{proof}
First let us show that in the first octant for any given point $p_0=(x_0,y_0)$ on the circle at the point $p_1=(x_0+1,y_0-1)$ must also be in the circle as $|p_1|<|p_0|$. When $p_0$ is in the circle and in the first octant segment such that $x<y$ it follows that the magnitude of $p_1$ is less than $p_0$ as $|p_0|-|p_1| = (x_0^2+y_0^2)-(x_1^2+y_1^2) = (x_0^2-x_1^2) + (y_0^2-y_1^2) = (x_0^2-(x_0+1)^2) + (y_0^2-(y_0-1)^2) = (x_0^2-x_0^2 - 2x_0 - 1) + (y_0^2-y_0^2 + 2y_0 -1) = 2y_0-2x_0-2 \geq 0$. Since $p_1=(x_0+1,y_0-1)$ should be in the circle if $p_0=(x_0,y_0)$ is on the circle we have that any chain code for the circle in the first octant does not contain a subword ``22'' as it will lead to contradiction. Moreover in any subword a single symbol ``2'' should be followed by ``0'' and can be substituted by $1$.
\end{proof}

As only the first octant segment is considered the solutions can be represented using strings consisting of only $0$ and $1$ where $0$ is positive motion along the $x-axis$ and $1$ is positive motion along the $x-axis$ and negative motion along the $y-axis$. 

Indeed in the generation of such chain codes the following method shall be used which is a modification of the algorithm that is given in \cite{Bhowmick20082381}. Starting with the equation for the circle $x^2+y^2=r^2$ where it is known, by definition, that $r^2\in \Zed$. Now as solutions of the first octant are required and under the assumption that, without loss of generality, the centre of the circle is at the origin it is now possible to replace $y^2$ with $(r'-k)^2$, where $r' = \lfloor r \rfloor$, such that when $k=0$ and rearranging for $x^2$ such that $x^2=r^2-r'^2-k^2+2r'k$ the solution to the circle equation gives $x^2= r^2-r'^2$ where $r'\leq r$, as such the length of the first solution of such an equation is $\sqrt{r^2-r'^2}$. Successive values of $k$ give solutions to the length in the $x$-axis of that particular maximal Pythagorean triangle. All such Pythagorean triangles will give the final solution to the circle. 

 Here it is necessary to start with some terminology to say more precisely which parts of the chain code are considered to represent specific components in the resulting polygons. It will become much clearer later on why such distinctions are required as a single side of the polygon may differ in its chain code categorisation. The first of these categories of chain code, the most basic, is given here.

\begin{definition}
A {\bf chain code segment} is a subword of a chain code which is either a word $0^{|s_0|}$
or $10^{|s_i|-1}$.
\end{definition}

A discrete circle, $u$, is composed of chain code segments where each segment is represented as $u_i$, i.e $u=u_0u_1u_2...u_n$. Chain code segments can be expressed using a power notation for the number of repetitions e.g $100100$ may be rewritten as $(100)^2$ and if $u=u_0u_1u_2u_3= (00)(100)(100)(10)=(00)^1(100)^2(10)^1$.

It is possible to map the chain code, $u$, in $Octant(1)$, to the others octants through a function that maps the alphabet of the chain code $\{0,1\}$, to code to a new code in another octant segment. A chain code of a circle in an $Octant(n)$ can be constructing from an $Octant(1)$
first by applying the following mapping $\{0,1\} \rightarrow \{n-1 \mod 8, n \mod 8 \}$, where $0\leq n \leq 7$.
Then for the $Octant(n)$, where $n \equiv 1 \mod 2 $, or $n$ is odd, the code and the mapping itself must be reversed such that $\{0,1\} \rightarrow \{n \mod 8, n-1 \mod 8 \}$. For the rest of this paper we will only consider chain codes in the first octant.

\begin{proposition}
A circle in the first octant can be chain coded as follows:
\begin{center}
$0^{|s_0|} 10^{|s_1|-1} 10^{|s_2|-1} \ldots 10^{|s_i|-1} \ldots 10^{|s_{r'^2/2}|-1}$
\end{center}
where,
\begin{center}
$s_i=\{x |r^2 - r'^2 + 2(i-1)r' - (i-1)^2 < x^2 \leq r^2 - r'^2 + 2ir' - i^2, x \in \Zed \}$
\end{center}
$r$ is the radius and $r'$ is the floor of the radius.
\end{proposition}

\begin{proof}
It follows from \cite{Bhowmick20082381} where such a method of chain coding is introduced.
\end{proof}

\begin{proposition}
 Chain code segments, following the schema $$0^{|s_0|} 10^{|s_1|-1} 10^{|s_2|-1} \ldots 10^{|s_i|-1} \ldots 10^{|s_{r'^2/2}|-1}$$ have the following constraints to their lengths in the discrete circle: $$\lfloor \frac{s_{i}-1}{2}\rfloor-1 \leq s_{i-1} \leq s_{i}+1\ .$$
\end{proposition}
\begin{proof}
It, again, follows from \cite{Bhowmick20082381} where this property is proved.
\end{proof}


As previously noted terminology must be introduced to distinguish the varying parts of the chain code and in order to talk about what they represent. In order to characterise the discrete circle two notions are introduced, \emph{Line Segments} and \emph{Gradients of Line segments}.The following definition for the gradient gives a method for extracting such information from the series of digits that construe the line as a word.

\begin{definition}
A {\bf gradient}, $G(u)$, of a chain code subword in $\{0,1\}$, $u$, is given by $G(u)=\frac{\#_1(u)}{|u|}$, where the function $\#_1$ returns the number of $1's$ found in the chain code and $|u|$ is the length of word $u$. 
\end{definition}

Line segments are in reference to actual lines that these sections of chain code would represent on the plane. That is if one were to draw a line and then attempt to translate this line in to a chain, under the assumption that the gradient of the line is rational, for reasons which can be seen in the definition of a gradient, it would be done via a periodic series of chain code segments. Taking a single period of this chain all the information that is required to identify and reconstitute this line to any length is now known. In this way the name may interpreted quite literally as a segment of a line in chain coded form such that any repetition of this object produces a line. With this in mind it will naturally be of benefit to be able to know a little more about the line that has been chain coded as a line segment. 

\begin{definition}\label{lineSeg}
A {\bf Line segment}, $l_i=u_ju_{j+1}...u_{j+n}$, is the shortest contiguous subsequence of chain code segments which maintain an increasing gradient such that $ G(l_i)<G(l_{i+1})$, $\forall i>0$, \tn{where the last chain code segment in $l_i$, $u_j+n$, is contiguous to the first chain code segment in the next line segment, that is, $l_{i+1}=u_{j+n+1}u_{j+n+2}...u_{j+n+m}$.}
\end{definition}

\begin{figure}[htp]
\centering
\includegraphics[scale=0.30]{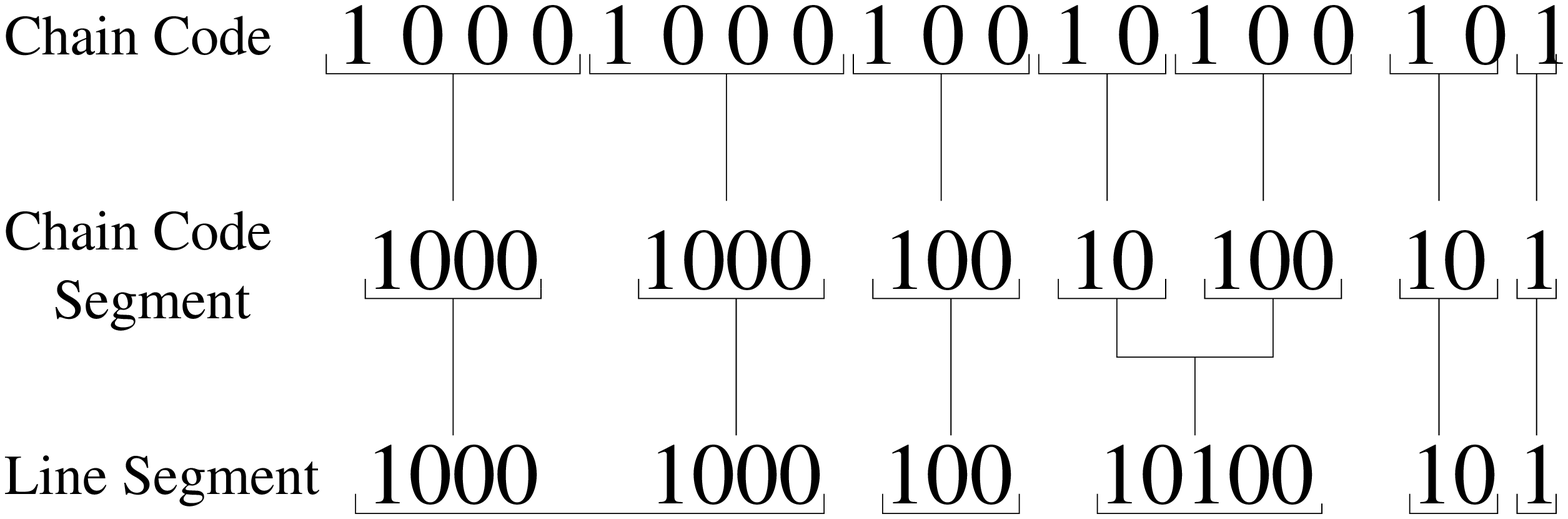}
\caption{Showing the the relationship between chain codes, chain code segments and line segments.}
\end{figure}

As, ultimately, all of the discrete discs that are discussed here, and, more importantly, that are possible to construct, are polygons, such polygons are composed of line segments which represent the sides of the convex shape on the plane. Further, when categorising the repeated broadcast of such discrete discs, later discussed as composition, the shapes that are generated are known to be similar, in the strict geometric sense. A definition for chain codes, translating from the geometry of the Reals, is given.

\begin{definition}
For any two chain codes, $u,v$, comprising line segments such that, $u=l_0l_1...l_m$ and $v=l'_0l'_1...l'_n$, for $m=n$, given by Definition~\ref{lineSeg}, are similar (geometrically) if there is a constant $k \in \Nat$ such that $u = l'\mathrlap{_{0}}^{k} l'\mathrlap{_{1}}^{k} ...l'\mathrlap{_{n}}^{k} $.
\end{definition}

In this way it is possible to say that two shapes are similar in a geometric sense such that they contain the same number of distinct line segments but the number of each of these line segments is different by some constant. 

With these definitions stated it is possible to begin the categorisation of the set of discrete circles with the first observation which extends an already known notion about the chain code of the discrete disc. Following \cite{Bhowmick20082381} it is known that in the chain code $u$ any chain code segments with increasing lengths, such that $|u_i|<|u_{i+1}|<...<|u_{i+n}|$, may increase it by at most $1$ from the length of a previous segment, i.e. $|u_{i+1}| \leq |u_i|+1$ for all $i$.

\begin{figure}[htp]
\centering
\includegraphics[scale=0.5]{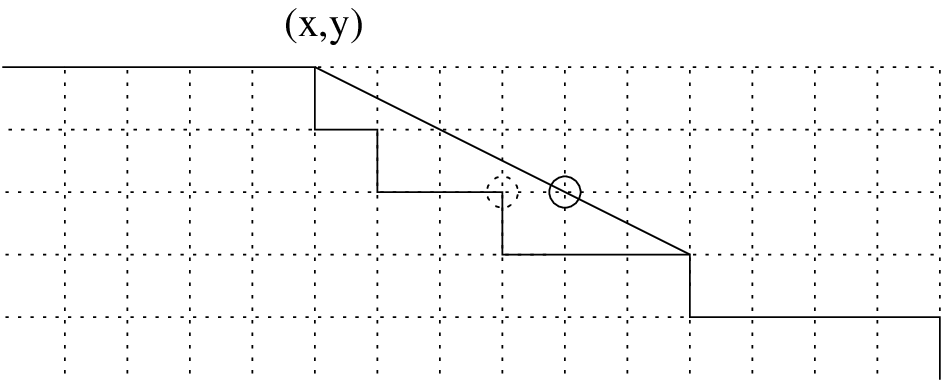}
\includegraphics[scale=0.5]{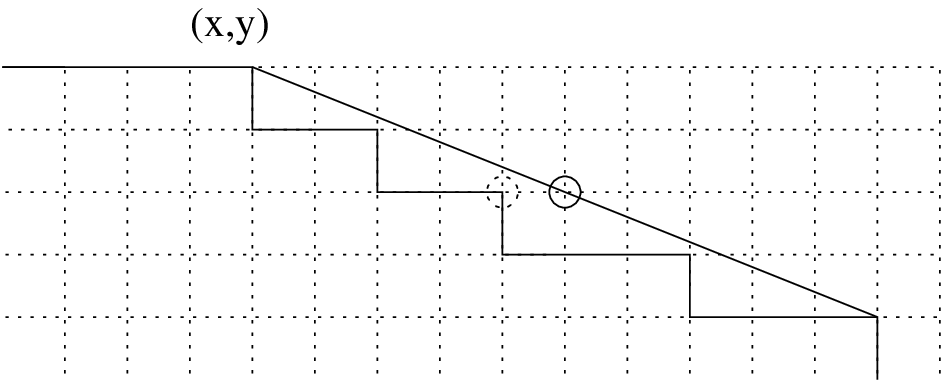}
\includegraphics[scale=0.2]{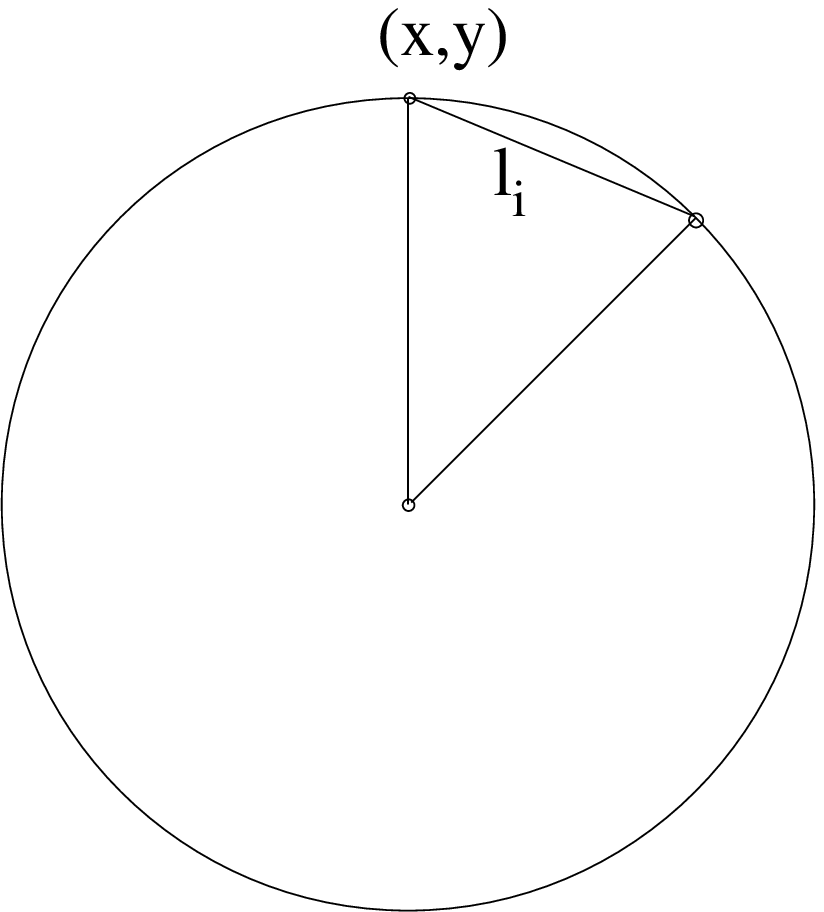}
\caption{(Left and Centre) Showing the initial point $(x,y)$ the chord and the point on the chord, the solid circle. (Right) Showing a chord on a circle from $(x,y)$ }\label{existence}
\end{figure}

By direct construction it is easy to check that a line segment can be of the form $10^n$ or $10^n 10^{n+1}$. However we can show that no line segment may be composed of more than two chain codes which increase in length by one. 

\begin{lemma}\label{lineLength}
No line segment can be of the form $10^n 10^{n+1} 10^{n+2} \ldots 10^{n+i}$, where $i>1$.
\end{lemma}

\begin{proof}
Let us first show that any line segment which is part of the discrete circle cannot be of the form $10^n 10^{n+1} 10^{n+2}$ for any $n>1$. We will prove it by contradiction. Assume that the line segment $l_i=10^n 10^{n+1} 10^{n+2}$ with the gradient $1/n+1$
starting at a point $p$ with coordinates $(x,y)$ and finishing at $(x+3n+6,y-3)$. Therefore a point $(x+2(n+2),y-2)$ , i.e. a point which can be reached by a code $10^n 10^{n+2}$ from a point $(x,y)$, is above the discrete circle. By the definition of the chordal property of the circle any line that joins two points on the circle must bound, within the triangle formed from the two points on the circle and its centre, points which are again within the circle. On the other hand a point $(x+2(n+2),y-2)$, shown circled in Figure~\ref{existence}, belongs to a chord between end points of a chord $(x,y)$ and $(x+3n+6,y-3)$ and therefore should be within a discrete circle. So it is not possible to have a line segment with the following chain code $10^n 10^{n+1} 10^{n+2}$. The same argument holds for the general case $10^n 10^{n+1} 10^{n+2} \ldots 10^{n+i}$, where $i>1$ since the extension of the line segment still keep the point $(x-2,y+2(n+2))$ above the discrete circle and on the other hand this point will be in the triangle formed by a centre of the circle and two end points of the line segment since its gradient is equal to $\frac{i+1}{n(i+1)+(i+1)+i(i+1)/2}=\frac{1}{n+1+i/2} \leq \frac{1}{n+2}$, where $i \geq 2$. 
\end{proof}

Further to this weak restriction, that the chain codes may not be monotonically increasing or indeed non-decreasing, it is possible to give a more definite generalisation of the structure of the chain codes. This structure, which must be adhered to for all discrete discs, is given here.

\begin{theorem}\label{lineSegFor}
Any line segments on the discrete circle with non-negative gradients
should be in one of the following forms 
$(10^{n})^*$, $(10^{n})(10^{n+1})^*$, $(10^{n})^*(10^{n+1})$.
\end{theorem}

\begin{proof}
Following Lemma~\ref{lineLength} we can restrict number of cases
since any concatenations of more than two chain codes which increase in length by one
may not be part of a line segment.
The line segments are sub-words of a chain code with increasing gradients, so if a subword
$(10^{n})^m$ is surrounded by any chain code segments $10^{n_1}$ and $10^{n_2}$, where 
$n_1, n_2 \ne {n \pm 1}$ , then $(10^{n})^m$ is a line segment.
In the case were $n_1, n_2 = {n \pm 1}$ a subword $(10^{n})^m (10^{n+1})^k$
can be surrounded by only chain codes $10^{n_3}$ and $10^{n_4}$, where 
$n_3 \ne n - 1$ and $n_4 \ne n +2$ (by Lemma~\ref{lineLength}).
We will show now that the only line segments of the form
$(10^{n})^m (10^{n+1})^k$ can be where either $m=1$ or $k=1$.
The proof is similar to Lemma~\ref{lineLength}. 
Let us first show that the line segment from a point $(x,y)$
cannot have the following chain code $(10^{n})^2 (10^{n+1})^2$.
If we assume that such chain code may correspond to a line segment 
such that a point $(x-2, y+2n+3)$, shown in the centre diagram in Figure~\ref{existence}, is above the chain code (i.e. our of a circle) and at the same 
time is on a chord between points $(x,y)$ and $(x-4, y + 4n +6)$, so should 
be in a discrete circle. Extending the chain code $(10^{n})^2 (10^{n+1})^2$
from the left by $(10^{n})^{m_1}$ or from the right by $(10^{n+1})^{m_2}$
for any $m_1, m_2 >0$ will not change the property of the point $(x-4, y + 4n +6)$.
So it can be in one of the following forms either $(10^{n})^m 10^{n+1}$ or $10^{n} (10^{n+1})^k$.
\end{proof}

\section{Iterative Composition}

Iterative composition of discrete circles may create different polygons, which later can also be used for forming
non-convex shapes.  The following definitions, and further characterisations, shall be required in order to discuss the properties of such composition and develop the proofs. 

\tn{
\begin{definition}
A discrete disc\index{discrete disc} is composed of points in the $\Zed^2$ set $\zeta^{r^2} =\{(x,y) | x^2+y^2\leq r^2,\ x,y\in \Zed,\ r\in \Real \}$.
\end{definition}

\begin{definition}\label{composition}
The composition of $l$ digital discs is defined as $\zeta^{r_0}\circ \zeta^{r_1}\circ...\circ\zeta^{r_l}$ for $r_i\in \Nat$, where, $\zeta^r \circ \zeta^{r'} = \{a+b|a\in \zeta^r,\ b\in\zeta^{r'}\}$ and is equivalent to $H(\mathcal{A}_l)$ where $\mathcal{A}=r_0,r_1,...,r_l$ as given in Definition~\ref{convHull} and Definition~\ref{rNeighbours} respectively.
\end{definition}}

\begin{definition}
The composition $u\circ v = w$ represents the composition of the chain codes in the first octant for their respective discs, $\zeta^u$ and $\zeta^v$, which results in $w$ the chain code for the first octant of the resultant disc of $\zeta^u\circ \zeta^v$. This method can be seen in Figure~\ref{exis}.
\end{definition}

\begin{figure}[htp]
\centering
\includegraphics[scale=0.8]{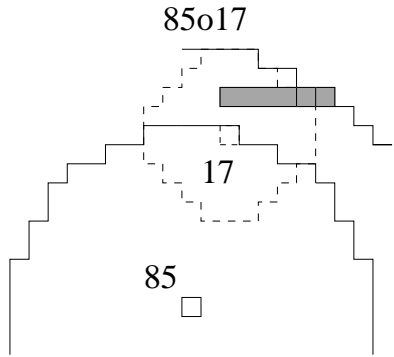}
\caption{Showing the composition of squared radii $Disc(85)\circ Disc(17)$ where $Disk(85)$ is shown as a solid line and $17$'s first composition is shown dashed. The hatched area represents all contributions from the previous compositions.}\label{exis}
\end{figure}




\begin{example}
Iterative composition of discrete circles\index{discrete circles} may create different polygons. For example, in the following picture, Figure~\ref{91626}, there are two differing cases, in one case we iteratively apply the digital circle with squared radius $9$ resulting in the equilateral octagon. In the second case squared radii $16$ and $26$ are periodically applied starting with $16$.
\end{example}

\begin{figure}[htp]
\centering
\includegraphics[scale=0.20]{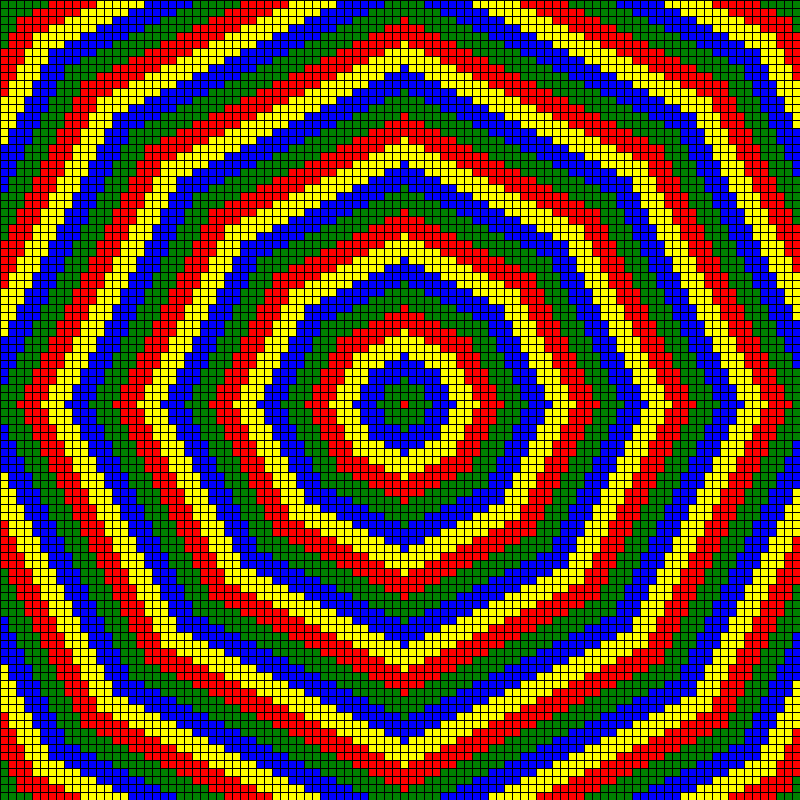}
\includegraphics[scale=0.20]{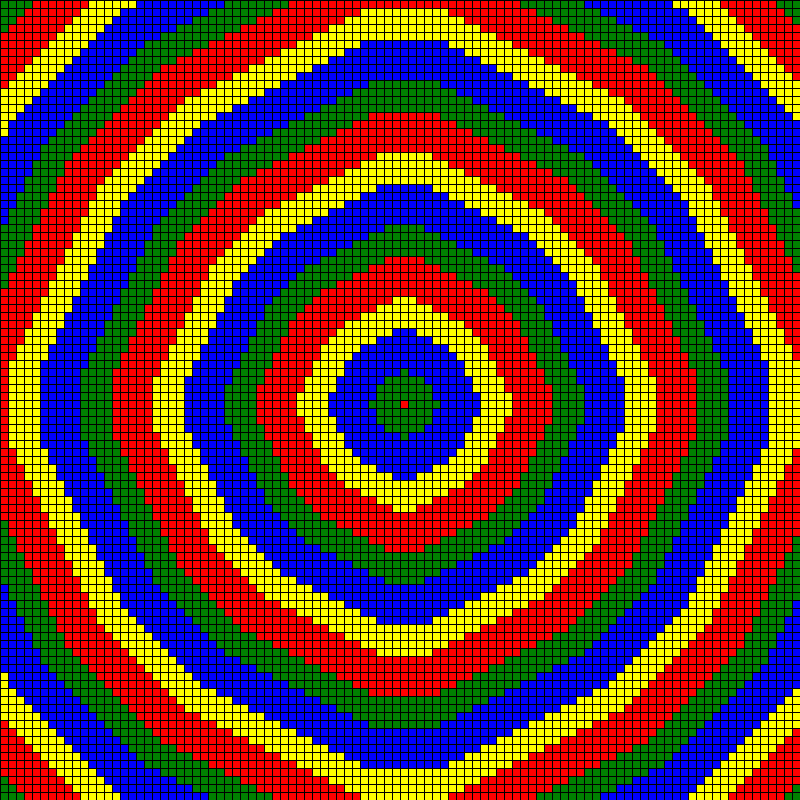}
\caption{The above figures illustrate (left) the constant iteration of the discrete disc of squared radius $9$ and (right) the alternation between the transmission of squared radii $16$ then $26$. }\label{91626}
\end{figure}

From these definitions it is now possible to describe a naive algorithm, and derive its correctness, for the composition of any number of chain codes that represent discrete discs. The result is again a convex polygon that represents the combination of the constituent parts. Such an algorithm merely systematically checks all possible combinations of the chain code and reports the largest that is found for that step. Later it shall be seen that if the chain code of the discrete disc is categorised further it is possible to give a simple linear time algorithm with a proof of correctness that is derived from the Minkowski sum along with a combinatorial proof.

\begin{lemma}\label{direct_composition}
Given two chain codes, $u=u_0u_1...u_n$ and $v=v_0v_1...v_m$, in form $\{0,1\}^*$ then $u \circ v=w=0^{|w_0|}10^{|w_1|-1}...10^{|w_{m+n-1}|-1}$, such that $$|w_k|=max(\{\displaystyle\sum\limits_{i=0}^{n} |u_i| + \displaystyle\sum\limits_{j=0}^{k-n} |v_j|\ |\ 0 \leq n \leq k \}).$$
\end{lemma}

\begin{proof}
The naive way of generating the composition $u\circ v$ is for all points on the chain code of the disc $u=0^{|u_0|}10^{|u_1|-1}...10^{|u_i|-1}...10^{|u_n|-1}$, to be the centre of the disc $v=0^{v_0}10^{|v_1|-1}...10^{|v_j|-1}...10^{|v_m|-1}$. This is equivalent to placing the centre of the chain code $v$ at every point defined by the chain code $u$ generating $w$. The maximal point in each chain code segment (i.e if $u_i=100$ then the coordinate of the final $0$) need only be considered, clearly covering all others. Let $u_i\circ^\prime v_j = w_{i+j}$ denote the centring $v$ at the coordinate reached by the maximal point of $u_i$ in which we consider the maximal point attained by chain code segment $v_j$ and represent a possible length of the chain code for $w$ at $w_{i+j}$. The maximal of all such combinations of lengths, those for which the sum of $i,j$ are equivalent, is required and is defined as the longest contiguous subsequence of length $k$ from $v=v_0v_1...v_j$ and $u=u_0u_1...u_i$ for all $i,j$ such that $i+j=k$ which represents $w_k$ for $0 \leq k < m+n$ i.e:
\begin{center}
$max  
\begin{pmatrix}
\vert u_0 \vert + \vert v_0 \vert \\
\end{pmatrix}
 = \vert w_0 \vert$
\hspace{1cm}
$max  
\begin{pmatrix}
|u_0| + |v_0| + |v_1| \\
|u_0| + |u_1| + |v_0| \\
\end{pmatrix}
 = \vert w_1 \vert$
%

\vspace{0.4cm}

$max  
\begin{pmatrix}
|u_0| + |v_0| + |v_1| + |v_2|\\
|u_0| + |u_1| + |v_0| + |v_1|\\
|u_0| + |u_1| + |u_2| + |v_0|\\
\end{pmatrix}
 = \vert w_2 \vert$

$$...$$

$max  
\begin{pmatrix}
|u_0| + |v_0| + |v_1| +...+ |v_k|\\
|u_0| + |u_1| + |v_0| + ...+ |v_{k-1}|\\
...\\
|u_0| + |u_1| + ... |u_i| + |v_0| + ...+ |v_{k-i}|\\
...\\
|u_0| + |u_1| + ... + |u_k| + |v_0|\\
\end{pmatrix}
 = \vert w_k \vert$

$$...$$

$|u_0| + |u_1| + ... + |u_n| + |v_0| + |v_1| +...+ |v_m| = |w_{m+n}|.$
\end{center}
\end{proof}

Such an algorithm allows a characterisation of the composition of discrete discs. Indeed, here, it can be shown that the composition of chain codes, which have been shown to be a word $W\in \{0,1\}^*$, is, again, a word, $W'\in \{0,1\}^*$.

\begin{corollary}
Given two chain codes $u$, $v\in \{0,1\}^*$. The composition $u\circ v\in \{0,1\}^*$
\end{corollary}

\begin{proof}
Each segment in $u$, $v$, is non-zero $\forall i\ |u_i|\neq 0, |v_i|\neq 0$. It is followed from Lemma~\ref{direct_composition} that by enlarging $k$, for $l(k+1)$, the previous set of solutions is extended by a chain code from $u$ or $v$ so $l(k) < l(k+1)$. Thus the composition will be in $\{0,1\}^*$.
\end{proof}

Also as a derivation from the algorithm an observation about the properties of the composition function $\circ$.

\begin{theorem}
The composition $\circ$ of two chain codes, $u$ and $v$ is commutative, i.e. $u \circ v = v \circ u$.
\end{theorem}
\begin{proof}
Following Lemma~\ref{direct_composition} it is possible to exchange $u$ and $v$ and still obtain the same result,
showing commutativity of composition of chain codes $u$ and $v$. The following two sets representing 
$k$-th level of a new chain code are equal as well as their maximums:
$$\{S|S= \displaystyle\sum\limits_{i=0}^{n} |u_i| + \displaystyle\sum\limits_{j=0}^{k-n} |v_j| ,0 \leq n \leq k \} =
\{S|S= \displaystyle\sum\limits_{i=0}^{n} |v_i| + \displaystyle\sum\limits_{j=0}^{k-n} |u_j| ,0 \leq n \leq k \} $$
\end{proof}

\tn{It is also notable that having shown commutativity chain code composition it may be possible that it is an Abelian group, where the identity element is simply the circle of radius $0$, associativity is derived from the algorithms ordering of line segments by gradient, which is naturally the same irrespective of the order it is carried out in, the inverse is simply the removal of one set of line segments from a chain code. However, it is currently unknown whether or not the operation is closed.}

Further, a proof of the ordering of the line segments in the discrete circle is here required in order to validate the proof of the composition theorem that is given as one of the main tools and results of this section.

\begin{lemma}\label{lineSegOrd}
Given line segments, $l_i$, of a form produced by the digitisation of a circle, which are composed of chain code segments of its first octant $a=10^{m-1}$ and $b=10^{m}$, then the ordering of their gradient for combinations of $a$ and $b$ is $$G(b)<...<G(ab^*)<...<G(ab)<...<G(a^*b)<G(a)\ .$$
\end{lemma}
\begin{proof}
The proof is a simple comparison of the gradients which shows, where $|a|=m$, $\frac{1}{m+1}<\frac{n}{(n-1)(m+1)+m}<\frac{n-i}{(n-1-i)(m+1)+m} <\frac{n-i}{(n-1-i)(m)+m+1}<\\<\frac{n}{(n-1)(m)+m+1}< \frac{n}{(n-1)m+m+1} $, where $i<n-1$, for $G(b)< G(ab^{n-1})<G(ab^{n-1-i})<G(a^{n-1-i}b)<G(a^{n-1}b)<G(a)$ respectively.
\end{proof}

Given the preceding validation about the convexity of the discrete discs it is now possible to employ the Minkowski sum to conclude and validate the composition of any number of discrete discs using the composition function.

\begin{theorem}\label{compoThe}
{\bf Composition Theorem.} Given two chain codes $u$ and $v$ which contain line segments 
$l^{u}_1 l^{u}_2 \ldots l^{u}_t$ and $l^{v}_1 l^{v}_2 \ldots l^{v}_{t'}$
with strictly increasing gradients.
The chain code of a composition $u \circ v$ can be constructed 
by combining line segments of $u$ and $v$ and ordering them
by increasing gradient.
\end{theorem}

\begin{proof}
One of the ways to prove the above statement is to use a similar result about the Minkowski sum \cite{schneider1993convex} of convex polygons in $\mathbb{R}^2$ and then to prove that it holds for the ordering of the segments of the digital circles in $\mathbb{Z}^2$. Here, a purely combinatorial proof is given in terms of chain codes, chain code segments and line segments. The above statement is proved by induction. The fact that the base case for the composition of the first two line segments holds can be seen by directly checking the expression from Lemma~\ref{direct_composition}. 

Assume that the statement of the lemma holds and the first $z$ line segments of $u \circ v$ were composed from a set $LS=\{l^u_1, l^u_2, \ldots , l^u_i, l^v_1, l^v_2, \ldots , l^v_j \}$ and contains $x$ chain code segments. Let us also denote a set with other line segments from $u$ and $v$ as $\bar{LS}$. Without loss of generality suppose that the last line segment that has been added is $ l^v_j$, so $G(l^v_j) \geq G(l')$, where $l' \in LS$ and $G(l^v_j) \leq G(l'')$, where $l'' \in \bar{LS}$. By adding the next line segment $l'''$ which will have a gradient larger or equal than all other line segments in $LS$ and smaller or equal than gradients of $\bar{LS}$ it must be shown that the extended chain code of $u \circ v$ is still correct, i.e. it satisfies to Lemma~\ref{direct_composition}. 

First of all note that adding a new part of the chain code would not change the previous $x$ layers. If a current maximum is in the form $$|u_p| + |u_{p-1}|+ ... + |u_0|+ |v_0|+|v_1|+ ... + |v_q|$$ then the next $\#_1(l''')$ sums will be extended by chain codes from $l'''$. If $l'''$ is a line segment in $u$ the gradient of $G(l^v_j)$ is greater then $G(l''')$ and, taking into account Lemma~\ref{lineSegOrd} and Theorem~\ref{lineSegFor}, it can be seen that $$|u_{p+1}| +|u_p| + |u_{p-1}|+ ... + |u_0|+ |v_0|+|v_1|+ ... + |v_q|$$ is a maximum within the following set $$\{|u_{p+1+shift}| + ... +|u_p| + |u_{p-1}|+ ... + |u_0|+ |v_0|+|v_1|+ ... + |v_{q-shift}|, shift \in \Nat \}$$ since the shift will represent removal of a larger value from the $v$ component and appending the smaller value from the $u$ component. 

Repeating the procedure and extending the sum by one value it can be seen that, again, $$|u_{p+2}|+|u_{p+1}| + |u_p| + |u_{p-1}|+ ... + |u_0|+ |v_0|+|v_1|+ ... + |v_q|$$ is a maximum within a set $$\{|u_{p+2+shift}| + ... +|u_p| + |u_{p-1}|+ ... + |u_0|+ |v_0|+|v_1|+ ... + |v_{q-shift}|, shift \in \Nat \}$$ following the same reasoning. Similar arguments can be applied for the case were $l''' \in v$.
\end{proof}

\begin{example}
Let us illustrate the composition of two chain codes corresponding to digital circles with squared radii 45 and 9 in the first octant. The chain code of the first octant is $\zeta^{45}= 0001$ and for $\zeta^{9} = 10$ under composition this yields $\zeta^{45}\circ \zeta^{9}=000101$.
\end{example}

From the preceding notions it is now possible to describe a linear time algorithm, which vastly out performs the exhaustive search of all combinations that has previously been given, which is able to form the composition of the chain codes of any two discrete circles. 
%
%
As Lemma~\ref{lineSegOrd} shows that all of the line segments will be in a non-increasing order it is possible to construct such line segments by counting the number of $0's$ after each delimiting $1$ combining those that increase in size with the preceding smaller chain code segment. Having found all line segments for both chain codes it is a simple case of, starting with the chain code with the line segment with the largest gradient, adding that element to a new chain code which is the composition of the initial two. Continue choosing the line segment with the largest gradient from either of the initial chain codes until there are no elements left in either of the original circles. The result of the algorithm is the composition of the two circles.
The following proposition gives a proof of the algorithm and the notion that it is linear time computable. 

\begin{proposition}
Two discrete circles chain codes can be composed into a single chain code in a linear time.
\end{proposition}
\begin{proof}
As Lemma~\ref{lineSegOrd} shows that all of the line segments will be in a non-increasing order it is possible to construct such line segments by counting the number of $0's$ after each delimiting $1$ combining those that increase in size with the preceding smaller line segment. Having found all line segments for both chain codes it is a simple case of starting with the chain code with the line segment with the largest gradient adding that to a new chain code which is the composition of the initial two.
\end{proof}

It is also now possible to state a more formal definition of `similar' such that it is possible to mathematically determine whether two objects are `similar', geometrically. The following algorithm is given as one application of such an algorithm and also hints at the notion of composing one discrete disc from other discrete discs which allows an insight in to the primality of discrete discs those that cannot be composed from any other set of discrete discs. 

\begin{theorem}
Given a finite broadcasting sequence of radii $R=(r_1,r_2,..., r_l)$ and
a convex polygon\index{convex polygon} $P$, it is decidable whether there are radii such that the chain coding of the composition of  is {\bf similar} to $P$.
\end{theorem}

\begin{proof}
The algorithm computes all line segments, and their corresponding gradients, for the chain codes of the set of digital disks with radii in, $\Real$, and the convex polygon, $P$. Each digital disk $r_i\in R$ can be represented as a vector with $k$ values, where $k$ is the cardinality of the set of gradients and each vector component corresponds to a particular gradient with an integer value standing for the number of line segments with this gradient in $R$. Finally we would need to solve a system of linear Diophantine equation over positive integers to check whether there is a set of factors for defined vectors that may match a vector for $P$ with another unknown factor.
\end{proof}

\section{Broadcasting Sequences and their Limitations}
Although the composition of circles gives us a large range of polygonal shapes,  an analysis of broadcasting sequences \index{broadcasting sequences}
shows that there certain limitations for such composition.
It's seen that by Lemma~\ref{lineLength} there are only three possible forms which the lines that comprise the discrete circle may take. Translating these in to their respective gradients gives the following three possible gradient forms 
$$G_1 = G((10^{n-1})^m) = \frac{m}{mn} = \frac{1}{n} $$ 
$$G_2 = G((10^{n-1})(10^{n})^m) = \frac{m+1}{n+m(n+1)}$$ 
$$G_3 = G((10^{n-1})^m(10^{n})) = \frac{m+1}{nm+n+1}\ .$$

The above gradients are reduced to a minimal form in order to elucidate the exclusivity of the gradients.

\begin{proposition}\label{posGrad}
It is only possible to express gradients in the first octant of a circle in a reduced form $\frac{1}{n}$, $\frac{a}{a(n+1)-1}$ or $\frac{a}{an+1}$ for $a,n\in \Nat$.
\end{proposition}

\begin{proof}
As $G_1$ is already in a minimal form it is obvious that this can only express those gradients, $g$, of the form $\frac{1}{n}$. For $G_2$ and $G_3$ the following method is employed, $G_2=\frac{m+1}{n+m(n+1)}=\frac{a}{b}$, where $a\bot b$ from which it follows that $m=a-1$ such that through substitution of $m$ in to denominator, $\frac{a}{a(n+1)-1}$ is arrived at. Similarly for $G_3$ by substitution the equation $\frac{a}{an+1}$ arises.
\end{proof}

It now becomes more clear of what cannot be expressed and the limitations inherent in the hulls of the broadcasting sequences.

\begin{proposition}
Not all rational gradients, $0\leq g\leq 1$, are expressed by the lines that comprise the discrete circle in the first octant. 
\end{proposition}

\begin{proof}
A counter example is given as proof. It is impossible to express any such rational of the form $\frac{5}{8}$. It is clear that it is not possible for $G_1$ to express such a fraction. For $G_2$ and $G_3$ the following suffices. $G_2=\frac{a}{a(n+1)-1}$ where $a=5$ such that $\frac{5}{5n+4}$ where there is no such $n\in \Nat$ such that $G_2=\frac{5}{8}$. Similarly for $G_3=\frac{a}{an+1}$ where $a=5$ and $\frac{5}{5n+1}$ there is no such $n\in \Nat$ such that $G_3=\frac{5}{8}$
\end{proof}

It also becomes clear from the composition theorem (Theorem~\ref{compoThe}) itself that it is not possible to generate any further gradients or new line segments through composition and in turn successive broadcasts of radii may not generate any polygons with chain codes that include gradients not previously present.

\begin{corollary}
The set of line segments, and as such gradients, that compose any discrete circle are closed under composition.
\end{corollary}

As noted here the shapes that are generatable are limited by their convexity as well as by the gradients of the lines that compose them. In the next section some of these limitations will be removed or relaxed using multiple broadcast sequences and an aggregation function with which to map all values to a single value.

\section{Reducing Restrictions Through Aggregation}

This section makes uses of the notion of aggregation to reduce the restrictions that are imposed by the continual composition, or simple construction of discrete discs. It can be shown that it is possible in certain cases, which shall be exposited both here and in further sections where it is employed to show the approximation of the astroid metric.

As all points may be labelled by their distance over some arbitrary modulo value an extension to the current work is proposed whereby two differing $r-neighbourhood$ sequences, $A$ and $A'$ are used to label the $\Zed^2$ lattice from the same point $p$. At any point, $p'$, there are now two labels such that $p'=(i,j)$ where by $k \equiv i \mod m$ for the sequence $A$ and $k' \equiv j \mod m$ for the sequence $A'$. Here two differing functions for the aggregation of values of $i$ and $j$ which define new shapes on the lattice and in turn new metrics. The new shapes defined by the lattice are not necessarily convex.

As all of the discrete disks which make up the labellings of the lattice are composed of discrete lines it is possible to analyse the effects of the combinations of disks by considering only the intersections of the lines that make up the disks. Consider a series of parallel lines expressed in the form $y_0=m_0x_0+c_0$ where $c_i\in \Zed$ is an arbitrary constant or offset, $m_i\in \Rat$ is any gradient permitted by the line segments of a discrete circle and $x,y\in \Zed$ are the usual Euclidean coordinates. These lines are such that each successive line is of a distance $w_i$ from the last which shall imitate the width between iterations of the discrete circles, for the first line segment this is equivalent to $\lfloor \sqrt{r^2}\rfloor$, and the discrete lines that they generate. All coordinates such that $m_0x_0+c_0+kw_0\leq y_0< m_0x_0+c_0+(k+1)w_0$ are labelled as $k_0$. A second set of parallel lines differing from the first are defined as $y_1=m_1x_1+c_1$ where all coordinates such that $m_1x_1+c_1+k_1w_1\leq y_1< m_1x_1+c_1+(k_1+1)w_1$ are labelled $k_1$. The intersection of these two areas is thus $(k,k')$. Increasing the offset of $k_0$ and $k_1$ to $k_0+1$ and $k_1+1$ naturally results in the labelling of the area of their intersection as $(k_0+1, k_1+1)$. The ordering of the tuple is not relevant to the functions that aggregate them. 

\tn{The first of the functions here have previously been studied in \cite{Geometric} with regards to geometric computations on the lattice. It is of particular interest in the Broadcasting Automata model where the notion of waves as observed in nature are used to control a large number of distributed automata. Being able to predict the resultant shapes that are formed by the transmission of waves is, naturally, largely advantageous in that it affords the ability to predict and manipulate the formations on the plane. Such formations can then be used for a variety of computational duties such as partitioning and geometric computation.

}

Two functions for aggregation will now be introduced with the relevant equations for resultant patterning of the lattice, where $m = 4$ and the values of $i,j \in \{0,1,2,3\}$.

\begin{definition}
The {\bf moir\'{e}} aggregation function is given in the table below.

\begin{minipage}[t]{0.5\linewidth}
\vspace{20pt}
{\it moir\'{e}(i,j)} \ \ \ \  =
\end{minipage}
\begin{minipage}[t]{0.5\linewidth}
\vspace{0pt}
\begin{tabular}{|l|l|l|l|l|}
  \hline
  $\oplus$ & 0 & 1 & 2 & 3  \\ \hline
  0 & a & b & c & b \\ \hline
  1 & b & a & b & c \\ \hline
  2 & c & b & a & b \\ \hline
  3 & b & c & b & a \\ 
  \hline
\end{tabular}
\end{minipage}
\end{definition}

\begin{definition}
The {\bf Anti-moir\'{e}} aggregation function can be expressed simply as addition over modulo 4 and is given in the table below.

\begin{minipage}[t]{0.5\linewidth}
\vspace{20pt}
{\it Anti-moir\'{e}(i,j)} \ \ \ \  =
\end{minipage}
\begin{minipage}[t]{0.5\linewidth}
\vspace{0pt}
\begin{tabular}{|l|l|l|l|l|}
  \hline
  $\oplus$ & 0 & 1 & 2 & 3 \\ \hline
  0 & a & b & c & d \\ \hline
  1 & b & c & d & a \\ \hline
  2 & c & d & a & b \\ \hline
  3 & d & a & b & c \\
  \hline
\end{tabular}
\end{minipage}

\end{definition}

\begin{proposition}
The gradients of the new lines formed by the moir\'{e} equation can be predicted using the equation $$\frac{w_o\cdot m_1- w_1\cdot m_0}{w_0-w_1}\ .$$ Variables refer to the, line gradient, $m_i$, and the width of the line, $w_i$.
\end{proposition}

\begin{proof}
The form of the moir\'{e} function is such that if the intersection of one area is labelled as $(k,k')$ then the next area that is labelled similarly will be of the form $(k+1,k'+1)$. Encoding this as the intersections of lines $y=m_0x+c_0+k\cdot w_0$ with $y=m_1x+c_1+k'\cdot w_1$ and $y=m_0x+c_0+(k+1)\cdot w_0$ with $y=m_1x+c_1+(k'+1)\cdot w_1$. All that is left is to find the gradient of the two intersections. Solving the first case, $x_0=\frac{c_1-c_0-k\cdot w_0-k'\cdot w_1}{m_0-m_1}$ and $y_0=m_0(\frac{c_1-c_0-k\cdot w_0-k'\cdot w_1}{m_0-m_1})+c_0+k\cdot w_0$ . The second, $x_1=\frac{c_1-c_0-(k+1)\cdot w_0-(k'+1)\cdot w_1}{m_0-m_1}$ and $y_1=m_ 0(\frac{c_1-c_0-(k+1)\cdot w_0-(k'+1)\cdot w_1}{m_0-m_1})+c_0+(k+1)\cdot w_0$. The gradient of the two points can be found, $\frac{y_1-y_0}{x_1-x_0}$, which results in the gradient of the labelling applied to the lattice.
\end{proof}

\begin{proposition}
The gradients of the new lines formed by the Anti-moir\'{e} equation can be predicted using the equation $$\frac{w_o\cdot m_1+ w_1\cdot m_0}{w_0+w_1}\ .$$ Variables refer to, line gradient, $m_i$, and the width of the line, $w_i$.
\end{proposition}

\begin{proof}
The form of the moir\'{e} function is such that if the intersection of one area is labelled as $(k,k')$ then the next area that is labelled similarly will be of the form $(k+1,k'+1)$. Encoding this as the intersections of lines $y=m_0x+c_0+k\cdot w_0$ with $y=m_1x+c_1+(k'+1)\cdot w_1$ and $y=m_0x+c_0+(k+1)\cdot w_0$ with $y=m_1x+c_1+k'\cdot w_1$. All that is left is to find the gradient of the two intersections. Solving the first case, $x_0=\frac{c_1-c_0-k\cdot w_0-(k'+1)\cdot w_1}{m_0-m_1}$ and $y_0=m_0(\frac{c_1-c_0-k\cdot w_0-(k'+1)\cdot w_1}{m_0-m_1})+c_0+k\cdot w_0$ . The second, $x_1=\frac{c_1-c_0-(k+1)\cdot w_0-k'\cdot w_1}{m_0-m_1}$ and $y_1=m_0(\frac{c_1-c_0-(k+1)\cdot w_0-k'\cdot w_1}{m_0-m_1})+c_0+(k+1)\cdot w_0$. The gradient of the two points can be found, $\frac{y_1-y_0}{x_1-x_0}$, which results in the gradient of the labelling applied to the lattice.
\end{proof}

Composition of this form is given as example in Figure~\ref{aggregate}. Here the two lines are shown on the lattice correspond to the gradients of the two differing aggregation functions.
The following proposition is now noted.

\begin{figure}[htp]
\centering
\includegraphics[scale=0.50]{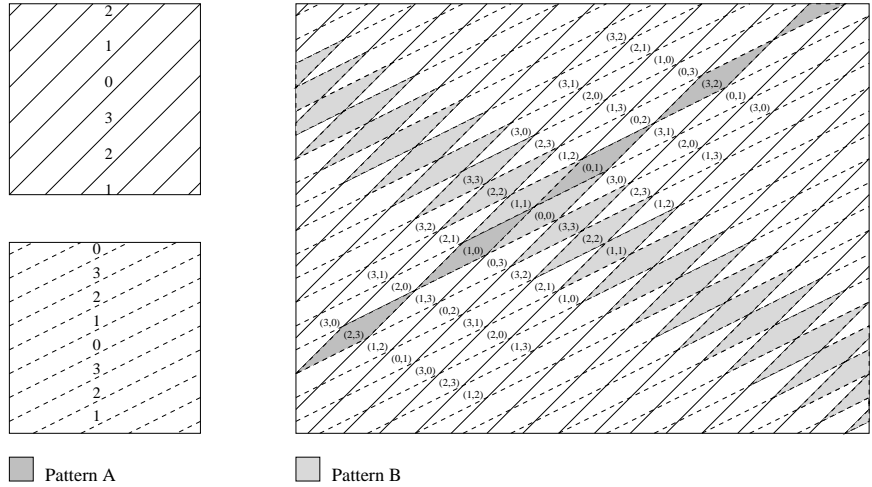}
\caption{The two forms of aggregation via the above functions. Here pattern A gives the line formed by the anti-moire function and the pattern b gives the line as formed by the moire function.}\label{aggregate}
\end{figure}

\begin{proposition}
For any two discrete lines on the $\Zed^2$ lattice it is possible to vary the observed gradients for resultant from aggregate functions by varying the width of the lines.
\end{proposition}

\begin{proof}
The proof is an observation of the equations for the gradients of the lines resultant from aggregation. Observing the results of two aggregating functions, $\frac{w_o\cdot m_1+ w_1\cdot m_0}{w_0+w_1}$ and  $\frac{w_o\cdot m_1- w_1\cdot m_0}{w_0-w_1}$, any alteration to the width of the line which represents results in a direct change in the gradient.
\end{proof}

It is now possible to formulate the following statement which shows the increase in expressivity that comes from using the composition of two broadcast sequences.

The number of lines formed by aggregation are not limited by the restrictions that are demonstrated in Proposition.~\ref{posGrad} where it is now possible to construct new lines through the varying of widths and indeed new polygons may now be formed with this technique such that they are non-convex.

\section{Formation of Polygons Through Aggregation.}

Having shown that it is possible to construct different gradients from line sections it is natural to now observe what happens when whole circles are intersected and the relationships that are formed by the aggregation functions that have been introduced. It is noted that a diagram for the formation of polygons in case of the moir\'{e} aggregation function is shown in Figure~\ref{moireComp}. Two digital discs, those with squared radii two and 25 such that the discs are $\zeta^2$ and $\zeta^{25}$, are composed generating, from the two convex polygons, and new, previously unreachable, polygon which is $non-convex$. 

\begin{figure}[htp]
      \includegraphics[scale=0.6]{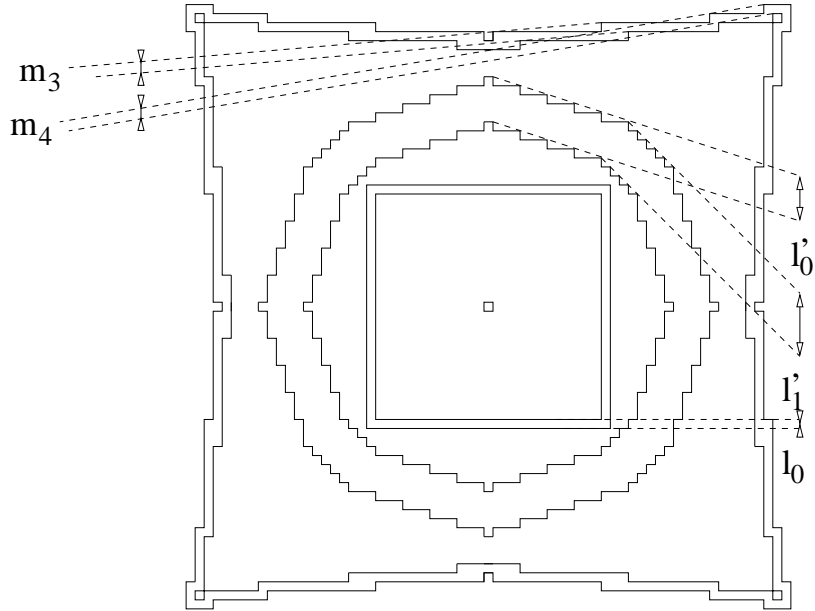}
      \caption{The above schematic depicts the broadcast of two discrete circles. The first, inner, circle (of squared radius two where $\zeta^2 = l_0$ with gradient $m_0$) and the second discrete circle (of squared radius $25$ where $\zeta^{25} = l'_0 l'_1$ with gradients $m_1$ and $m_2$ respectively) the outer construction shows the resulting moir\'{e} lines here labelled $m_3$ and $m_4$. Arrowed lines show line width measurements, the respective $w_i$, here, $w_0=1$, $w_1=5$, $w_2=7$, $w_3=\frac{5}{4}$ and $w_4=\frac{7}{6}$.}\label{moireComp}
\end{figure}

The following examples given in Figure~\ref{moantimo} also elucidate the differences between the two aggregation functions showing the non-convex polygon generated by the two digital discs that are represented by, $\zeta^2$ and $\zeta^9$. 

\begin{figure}[htp]
\centering
\includegraphics[scale=0.10]{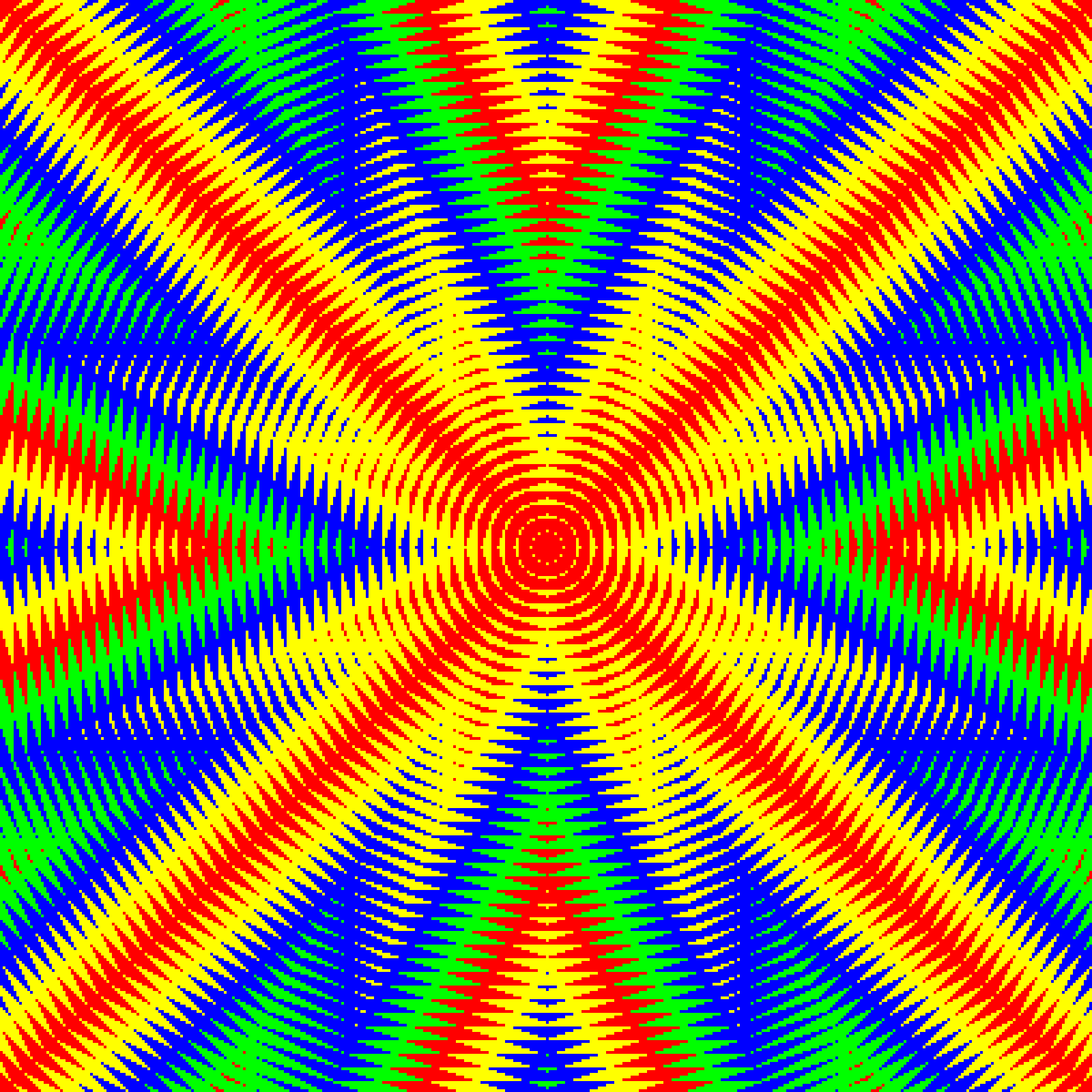}
\includegraphics[scale=0.10]{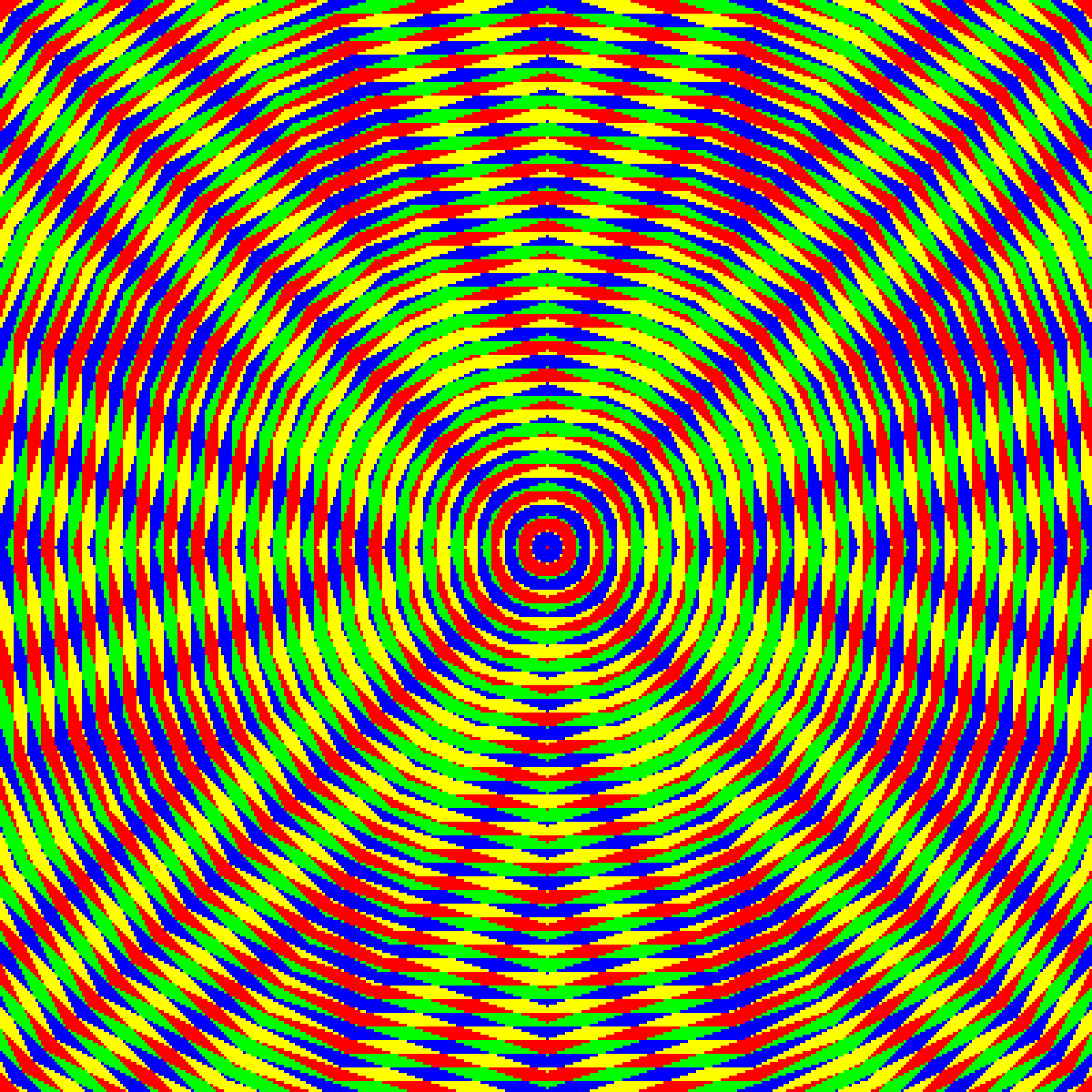}
\caption{Above gives an example of aggregation of two broadcast sequences, one of the discrete disc $ \zeta^{32}$ and the other $\zeta^{36}$, from a central point on the diagram and with the two aggregation functions (left) moir\'{e} and (right) Anti-moir\'{e}.}\label{moantimo}
\end{figure}

It is also possible to gain a better picture of the overall shapes produced by the anti moire equation by reducing the number of labels, and so the colours, further. This is done by a process of merging certain values or reducing the values over some modulo which in this example, Figure~\ref{antimocomb}, is two. The reduction is given by the following aggregation function:

\begin{minipage}[t]{0.5\linewidth}
\vspace{20pt}
{\it Anti-moir\'{e}-Mod2(i,j)} \ \ \ \ =
\end{minipage}
\begin{minipage}[t]{0.5\linewidth}
\vspace{0pt}
\begin{tabular}{|l|l|l|l|l|}
  \hline
  $\oplus$ & 0 & 1 & 2 & 3 \\ \hline
  0 & a & b & b & a \\ \hline
  1 & b & b & a & a \\ \hline
  2 & b & a & a & b \\ \hline
  3 & a & a & b & b \\
  \hline
\end{tabular}
\end{minipage}

\begin{figure}[htp]
\centering
\includegraphics[scale=0.2]{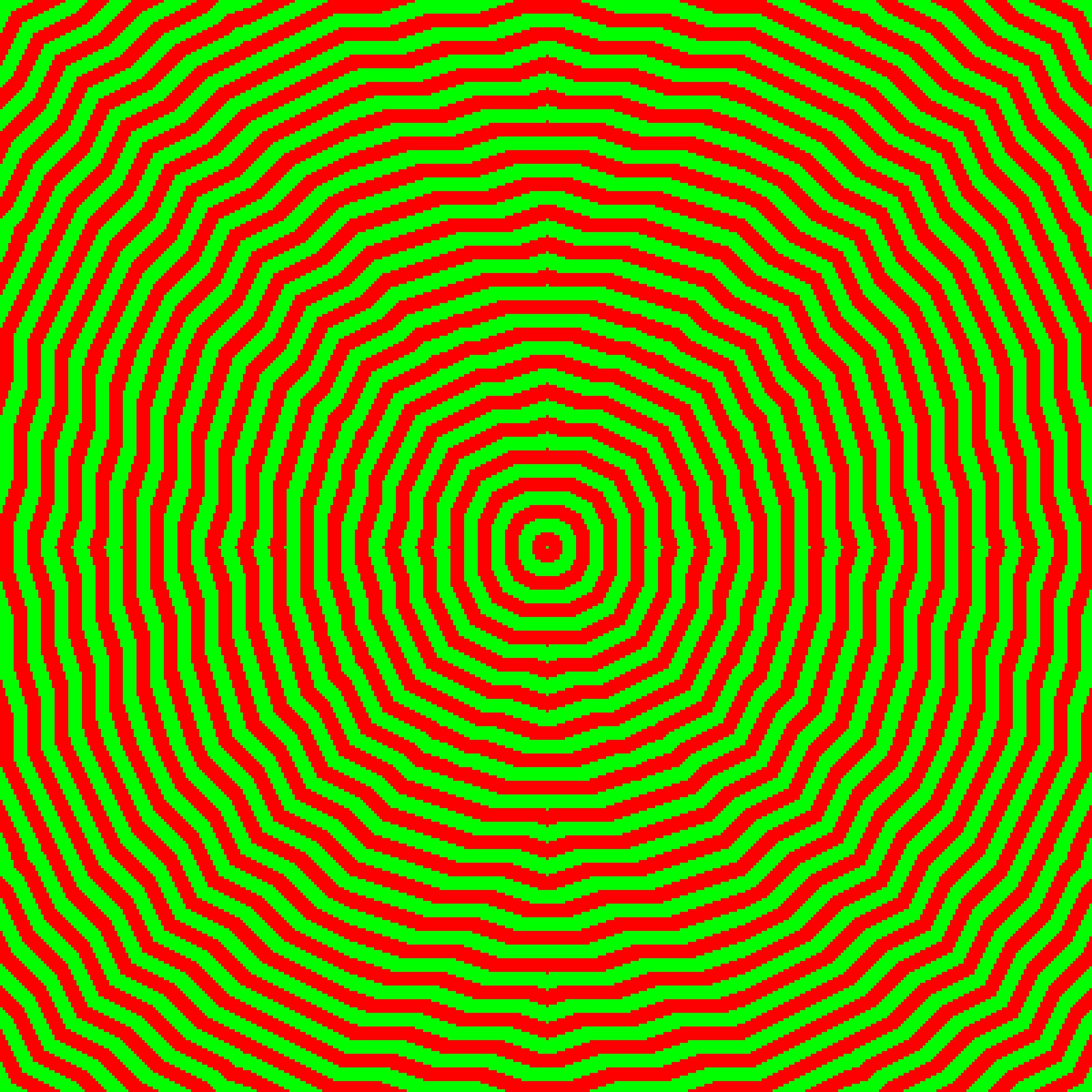}
\caption{The above figure depicts the anti moire pattern after the merging of two sets of two colours and initially generated by the discrete discs of squared radius $32$ and $36$ from the central point on the plane.}\label{antimocomb}
\end{figure}

\section{Approximating $L^p$ Metrics with Broadcasting Automata}\label{Lp_Metrics}

Previously, von Neumann and Moore neighbourhoods have been used to achieve an approximation of the Euclidean metric through some mixing of the neighbourhoods. This has been done in many different settings such as periodic, non-periodic combinations of the two neighbourhoods, regular and non-regular, i.e, hexagonal, triangular grids to which different sequences, i.e applying different definitions of neighbourhood, are applied, as well as a variety of methods for defining just how an approximation of Euclidean distance by neighbourhood sequences should be defined and measured, where most of these techniques discuss notions of digital circularity such as the isoperimetric ratio, perimeter comparisons, etc. In short this has been one of the main studies with regards to neighbourhood sequences. There is ultimately, as has been previously discussed, a large barrier to the extension of this body of work in the approximation of the more general $L^p$ metrics due to the impossibility of constructing any non-convex polygon from the composition of the two convex polygons which represent the Moore and von Neumann neighbourhoods. The astroid\index{astroid} is part of the `family' of $L^p$ metrics of which the Moore neighbourhood, $L^{\infty}$, the von Neumann neighbourhood, $L^1$, and the Euclidean metric, $L^2$ are the most well known. In general an $L^p$ space\index{$L^p$ metrics} is defined by the formula $\parallel x\parallel_p=(|x_1|^p+|x_2|^p+...+|x_n|^p)^{\frac{1}{p}}$ such that $\parallel x\parallel_p$ is the $p$-norm.

It is in this section that a new method for the generalisation of metric approximation\index{metric approximation} with broadcasting neighbourhoods shall be given and, using only two broadcasting radii, here given in terms of $r^2$, and the moir\'{e} aggregating function, explore the ability of this model to approximate the astroid and, as such, a new class of metrics outside of the reach of neighbourhood sequences. 

It has previously been seen that the methodology of combining two broadcasting sequences using an aggregating function can be used to extend the possible resulting polygons. In this section it shall be seen that this can be employed in the solution to a practical problem. The problem, in this case, is the approximation of $L^p$ metric, $L^{2/3}$, which forms an \emph{astroid} \cite{yates1974curves}. The astroid can be expressed by the equation, $x^{2/3}+y^{2/3}=r^{2/3}$, where $x$ and $y$ are those of the Cartesian plane and, as with the equation for the circle from which this equation is generalised, the $r$ is the 'radius' of the astroid.

A few preliminaries with regards to the astroid which are taken from \cite{yates1974curves} will be required to understand the methods proposed for comparing the approximation, via broadcasting sequence and aggregation, and the actual astroid. The astroid was discovered in Roemer in 1674 whilst searching for the best form of gear teeth. It is a hypercycloid of four cups and can be described by a point on a circle of radius $\frac{3}{4}\cdot a$ rolling on the inside of a fixed circle of radius $a$. The resultant shape has a perimeter, $L$, of $L=6a$ where $a$ is the radius of the outer fixed circle and is the maximal point reached by the astroid. The area, $A$, encapsulated is given by $A=\frac{3}{8}\pi a^2$. Finally, the point at which the $x$ and $y$ coordinates are equivalent on the perimeter of the astroid can be expressed as $\frac{r}{2}$. The equation is known to have applications in magnetism where the Stoner-Wohlfarth astroid curve separates regions with two minima of free energy density from those with only one energy minimum and is a geometric representation of the model of the same name \cite{Thiaville19985}.

\begin{figure}[htp]
\centering
\includegraphics[scale=0.32]{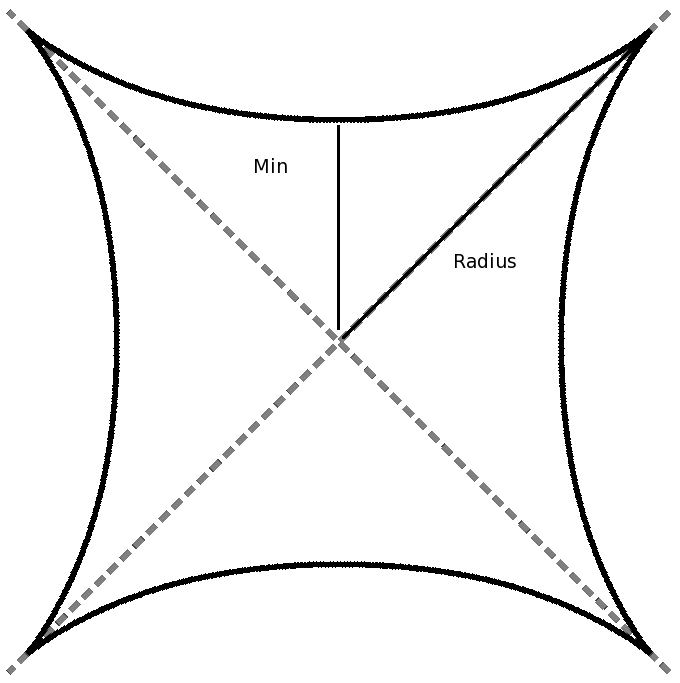}
\includegraphics[scale=0.4]{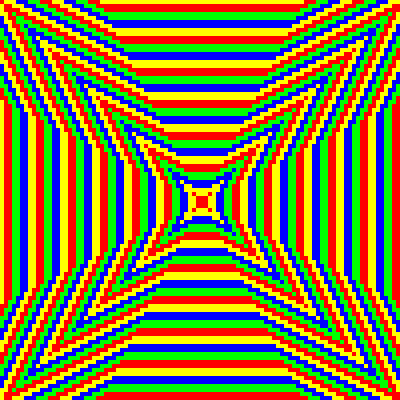}
\caption{The above figure depicts the astroid, here, rotated by $45^\circ$ (left) compared to its best, found, approximation with broadcasting sequences (right), an aggregation of the discrete discs of squared radius, $2$ and $5$. }\label{astroid}
\end{figure}

In \cite{springerlink:10.1007/117749388, Hajdu2004101} a number of methods for how well a particular neighbourhood sequence approximates the euclidean distance are given. The authors of \cite{springerlink:10.1007/117749388} consider any sequence, including those sequences which are non-periodic, of neighbourhoods with which to approximate the euclidean distance. In order to show how well, or poorly, the sequence approximates the euclidean distance it is compared to the euclidean circle through a variety of methods. One such method is through the use of the isoperimetric ratio or the noncompactness ratio. Here, the attempt is to measure $\frac{P^2}{A}$ where, $P$, is the perimeter, and, $A$, is the area. This measure is conjectured to be minimal for the circle where it is $4\pi$, however, this does not help when looking at non-convex shapes, such as those produced when approximating the an astroid, due to the measure being optimised for convex and symmetrical shapes. 

Other ways of approximating the euclidean circle are suggested such as a perimeter based approximation and area based approximation. With respects to are based approximation there are two techniques used. The first is that of the inscribed circle based approximation. This method attempts to find a sequence that generates the polygon, generated by the neighbourhood sequence, which is closest to the polygon having the given circle as the inscribed circle. The second of these methods is the covering circle based approximation such that polygon must be covered by the circle. More methods are given but rely on properties distinct to the circle and so are not discussed here. 

Further, in \cite{Hajdu2004101}, another measure of circularity is considered using three differing methods. They formulate three approximation problems whereby the aim is to find the neighbourhood sequence that minimises the error in each formulation. The problems may be informally described as the following: problem one requires finding the neighbourhood sequence that best minimises the size of the symmetric difference between some neighbourhood sequence at step $k$, given as $A_k$ and the Euclidean circle of radius $k$; the second problem attempts to find the sequence that best minimises the complement of the neighbourhood generated by $A_k$ with it's smallest inscribed circle; the third problem is a discretisation of the second problem where an $A_k$ must be found such that it minimises the complement with the largest discrete disc that can be inscribed within.

It is the second of these problems from \cite{Hajdu2004101} that shall be explored as a method for showing a best approximation of the astroid with combined broadcasting sequences. In this case the method is the most simplistic to calculate and also the least dependent on any of the properties that are only present in the circle. As such a more formal representation of the problem can now be posed, here, given as, $Area(H(f_k(A, B))\backslash G'_k\leq Area(H(f_{k'}(A, B'))\backslash G'_{k'})$, where, $f_k(A, B)$ is the $k$th polygon generated by the aggregation of the broadcasting sequences, $A$, which here will be the discrete disc, $r^2=2$, and the broadcasting sequence, $B,B'$, here a variable which is to be optimised to find the best approximation. There is a simple change to $G_k$, originally used in \cite{Hajdu2004101} to represent the circle of radius $k$, to convert it to $G'_k$ in that $G'_k=\{q\in \Real^2:L^{\frac{2}{3}}(0,q)\leq k\}$ as simple conversion of the metric from that of the euclidean distance, $L^2$, to the astroid distance, $L^{\frac{2}{3}}$.

The complexity of calculating these compositions for the purposes of optimisation means that only experimental data shall be given here as a proof of validity of the approximation of the concave metrics, which the astroid represents. As the astroids require a rotation by $45^\circ$ in order to match the polygon that is generated by the aggregation, $f(A, B)$, as defined before. In matching the point that is at the largest euclidean distance from the origin, which for simplicity, and without loss of generality, is considered the initial point from which all broadcasting occurs. This point is matched to the same point on some polygon on some composition, $f(A, B)$, and the complement of the areas compared. The following table is produced with this method.

The function of aggregation must also be defined. Here the choice is moir\'{e} aggregation without the modulo restriction, although, such a restriction is retained in the images for simplicity. Such a function can now be defined as, $f(A_i,B_j) = |i-j|$ for the $i$th and $j$th iteration of the sequence $A$ and $B$ respectively. For the function $f_k(A,B)$ where there exists some $A_i$ and $B_j$ such that $|i-j|=k$.

The min point for the $L^{\frac{2}{3}}$ is calculated by $\frac{r}{2}$ where $r$ is the radius of the astroid, the $r$ in $x^{\frac{2}{3}}+y^{\frac{2}{3}}=r^{\frac{2}{3}}$. The $B$ is the second broadcasting sequence, here, given as the $r^2$ of the disc. Images for each of the resultant, aggregated images are given in Table~\ref{astpic1} and Table~\ref{astpic2}.

\begin{table}[position specifier]
\centering
\begin{tabular}{| c | c | c | c | c | c | c | c | c |} \hline
  \multicolumn{3}{|c|}{Astroid - $L^\frac{2}{3}$} & \multicolumn{3}{|c|}{moir\'{e} - $f_k(A,B)$} & \multicolumn{3}{|c|}{}\\ \hline
  Radius & Min & Area & Radius & Min & Area & B & k & Complement \\ \hline
  17 & 8.5 & 340 & 17 & 11 & 461 & 5 & 5 & 121  \\ \hline
  17 & 8.5 & 340 & 17 & 13 & 753 & 9 & 8 & 413  \\ \hline
  17 & 8.5 & 340 & 17 & 13 & 873 & 10 & 9 & 533 \\ \hline
  18 & 9 & 382 & 18 & 16 & 1121 & 13 & 11 & 739 \\ \hline
  18 & 9 & 382 & 18 & 14 & 1033 & 16 & 11 & 651 \\ \hline
  18 & 9 & 382 & 18 & 14 & 1001 & 17 & 11 & 619 \\ \hline
  17 & 8.5 & 340 & 17 & 15 & 1041 & 37 & 13 & 701 \\ \hline
  17 & 8.5 & 340 & 17 & 16 & 1141 & 45 & 14 & 801 \\ \hline
  17 & 8.5 & 340 & 17 & 16 & 1093 & 61 & 14 & 753 \\ \hline
  17 & 8.5 & 340 & 17 & 16 & 1181 & 82 & 15 & 841 \\ \hline
\end{tabular}
\caption{Illustrating the experiments done with regards to the approximation of the astroid.}\label{compAst}
\end{table}
\begin{table}[h!]
\centering
\begin{tabular}{c c}

\includegraphics[scale=0.35]{r5.png} 
&

\includegraphics[scale=0.35]{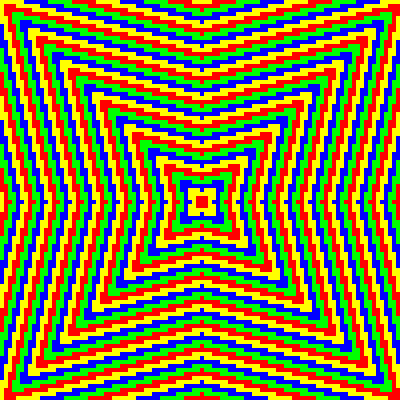}
\\

\includegraphics[scale=0.35]{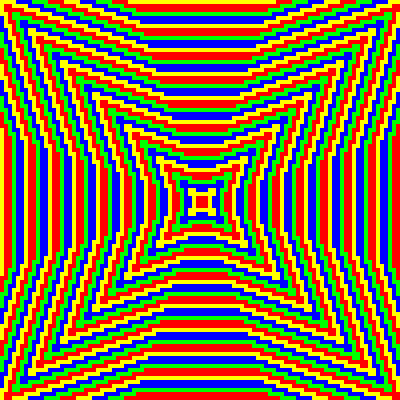}
&

\includegraphics[scale=0.35]{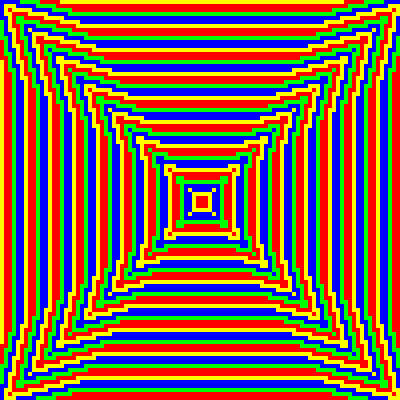}
\\
\end{tabular}
\caption{The above diagrams show the patterns generated by the aggregation of the two discs, where one is the disc of squared radius $2$ and the other is varied, in these diagrams (from left to right and line by line) there are discs of squared radius, $5$, $9$, $10$ and $13$.}\label{astpic1}
\end{table}

\begin{table}[position specifier]
\centering
\begin{tabular}{c c}
\includegraphics[scale=0.35]{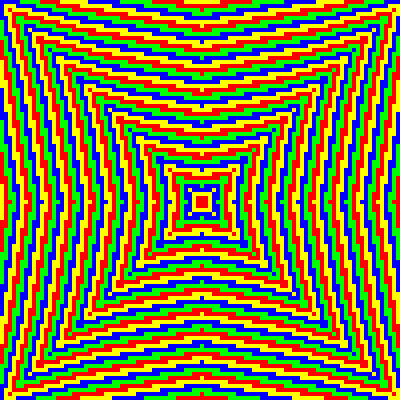}
&

\includegraphics[scale=0.35]{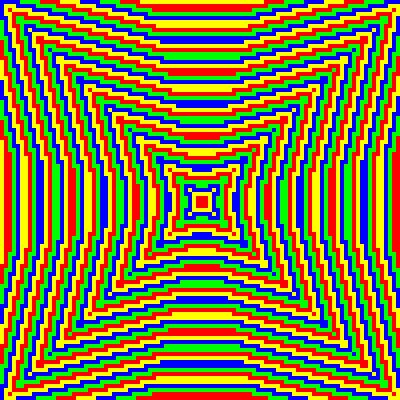}
\\

\includegraphics[scale=0.35]{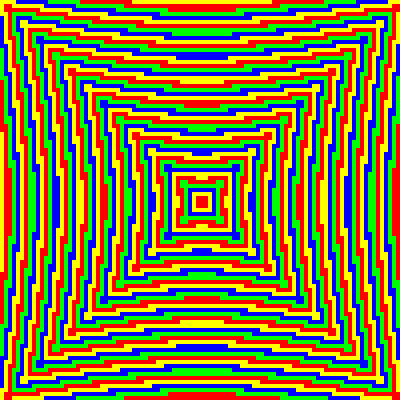}
&

\includegraphics[scale=0.35]{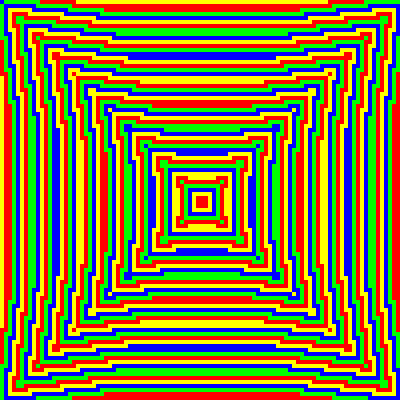}
\\

\includegraphics[scale=0.35]{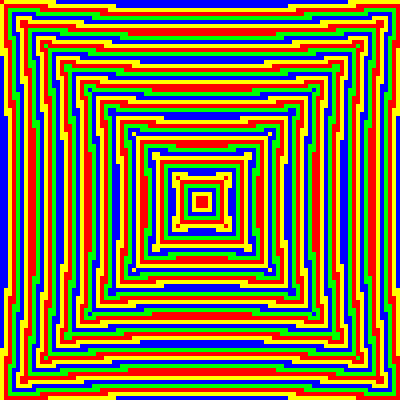}
&

\includegraphics[scale=0.35]{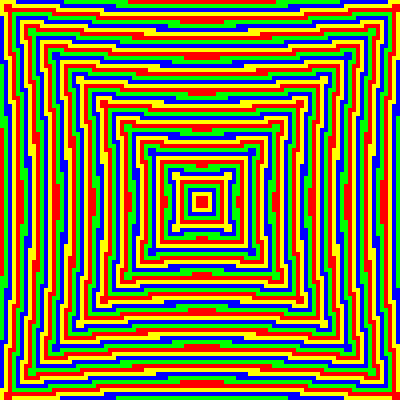}
\\

\end{tabular}

\caption{The above diagrams show the patterns generated by the aggregation of the two discs, where one is the disc of squared radius $2$ and the other is varied, in these diagrams (from left to right and line by line) there are discs of squared radius, $16$, $17$, $37$, $45$, $61$ and $82$.}\label{astpic2}

\end{table}


Tables~\ref{compAst} and ~\ref{compAst2} give a series of comparisons for $Area(H(f_k(A, B)))\backslash G'_k$. From this table it is possible to observe that the best approximation, from those constructed, though there is a trend towards worsening approximations as $B$ increases, that the simplest composition yields the best results. In this case this value for $B$ is the disc generated by the squared radius $r^2$ of $5$, this disc being constructed by the next smallest squared radius that constructs a new discrete disc. The following table now looks, again, experimentally, at how the approximation changes as $k$ increases. The table notes that the approximation weakens as it increases perhaps indicating a divergence between the astroid and the polygon generated by $f_K(A,B)$.

\begin{table}[position specifier]
\centering
\begin{tabular}{| c | c | c | c | c | c | c | c | c |} \hline
  \multicolumn{3}{|c|}{Astroid - $L^\frac{2}{3}$} & \multicolumn{3}{|c|}{moir\'{e} - $f_k(A,B)$} & \multicolumn{3}{|c|}{}\\ \hline
  Radius & Min & Area & Radius & Min & Area & B & k & Complement \\ \hline
  2  & 1    & 5   & 2  & 1     & 13  & 5 & 1      & 8 \\ \hline
  5  & 2.5  & 30  & 5  & 3     & 65  & 5 & 2      & 35 \\ \hline
  8  & 4    & 75  & 8 & 5      & 157 & 5 & 3      & 82 \\ \hline
  11 & 5.59 & 143 & 11 & 7     & 289   & 5  & 4   & 146 \\ \hline
  18 & 9    & 382   & 18 & 9   & 461   & 5  & 5   & 79 \\ \hline
  21 & 10.5 & 520   & 21 & 11  & 673  & 5  & 6    & 153 \\ \hline
  24 & 12   & 679   & 24 & 13  & 925& 5 & 7       & 246 \\ \hline
  27 & 13.5 & 859   & 27 & 15  & 1217& 5 & 8      & 358 \\ \hline
  30 & 15   & 1060   & 30 & 17 & 1549& 5 & 9      & 489 \\ \hline
  33 & 16.5 & 1283   & 33 & 19 & 1921& 5 & 10     & 628 \\ \hline
\end{tabular}
\caption{Illustrating the experiments done with regards to the approximation of the astroid where both $A$ and $B$ are fixed.}\label{compAst2}
\end{table}

\section{Pattern Formations and Periodic Structures in $\Zed^2$}

It is natural after observing the variety of effects that are the result of the application of an aggregation function to look at functions which are themselves shapes of some form. Here, functions of this form and their resultant patterns, imposed on the grid according to their application, are given. Such functions are only restricted by their symmetry, a result of the unordered nature of the tuples to be aggregated.

The first function, depicted in Figure~\ref{discreteDisc882636} (Left) takes the form of a discrete disc itself, in this case one which is also represented by the discrete disc of squared radius five. The functions here have also been increased in size and the size of the modulo for the labels has been increased. The use of discrete discs of squared radius 8 is important here as it constructs, in some sections of the lattice, a perfect reproduction of the shape given in the aggregating function. The following table describes the aggregating function and the details of the image.

\begin{figure}[htp]
\centering
\includegraphics[scale=0.12]{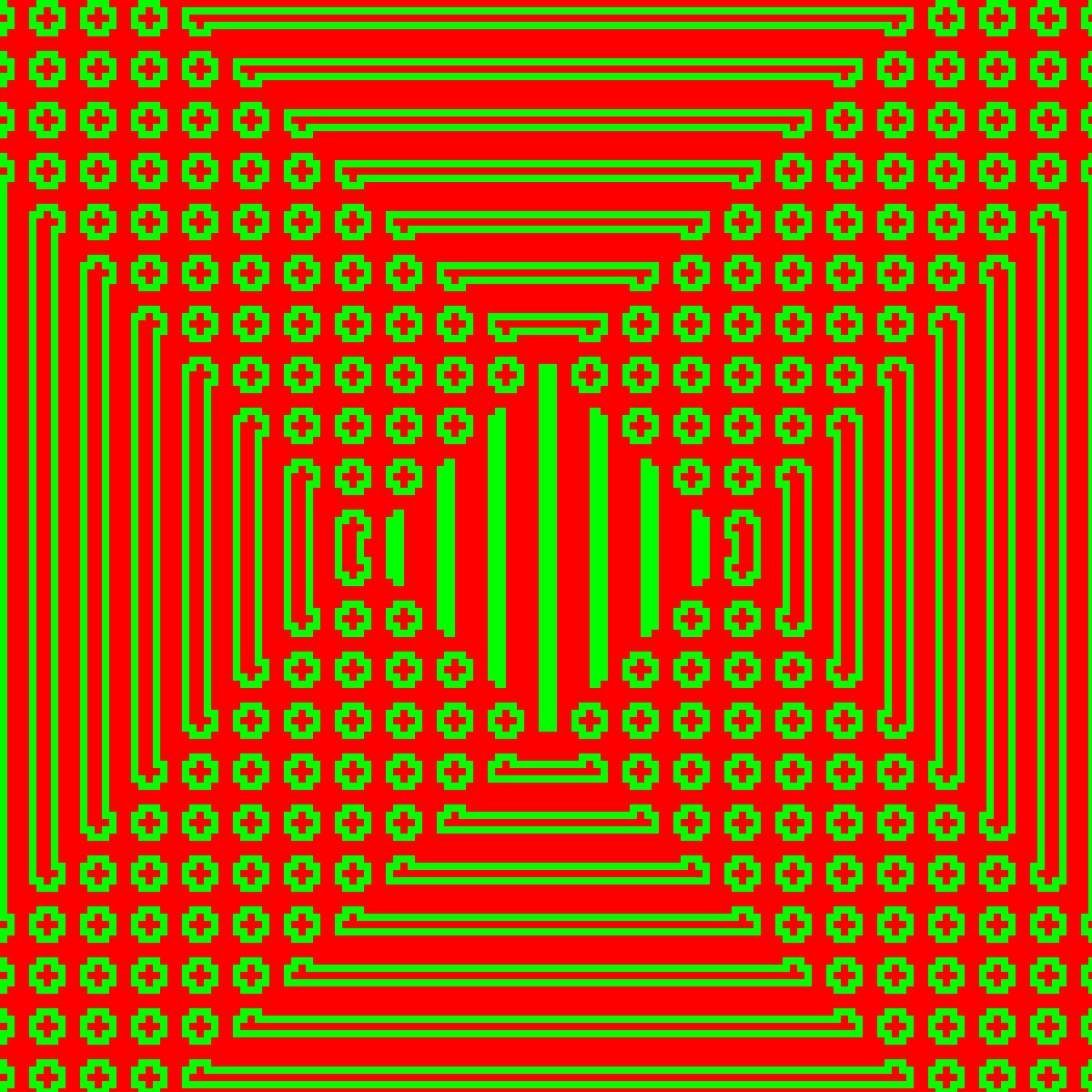}
\includegraphics[scale=0.12]{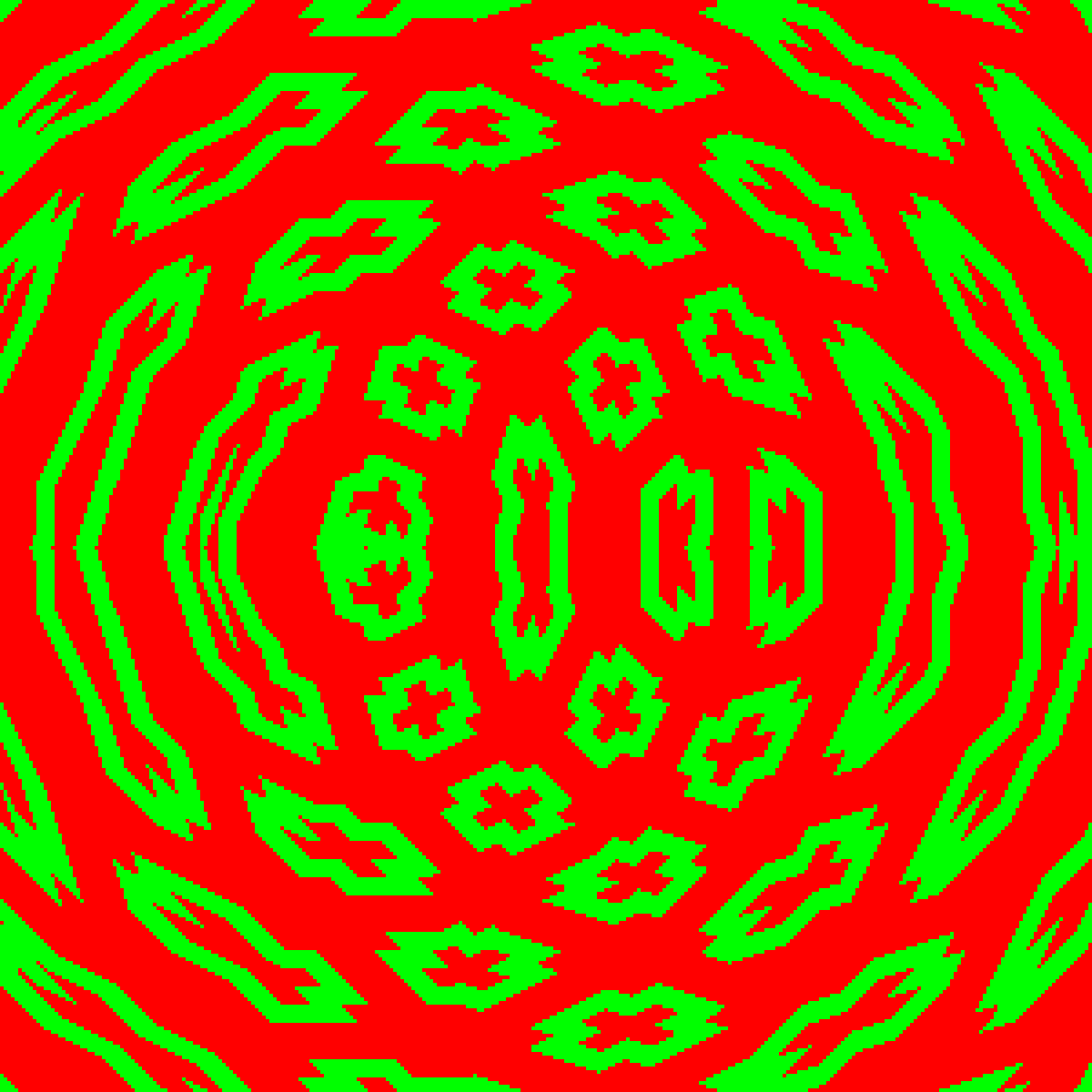}
\caption{(Left) Patterns generated by the aggregating function representing the discrete disc of squared radius five and the labelling of the lattice given by two broadcasting sequences of squared radius eight. (Right) Patterns generated by the aggregating function representing the discrete disc of squared radius five and the labelling of the lattice given by two broadcasting sequences, one of  squared radius 26 and the other of squared radius 36. }\label{discreteDisc882636}
\end{figure}

\begin{minipage}[t]{0.5\linewidth}
\vspace{0pt}
\begin{description}
\item[Array Size:] 300
\item[Centre 1:] (100,150)
\item[Centre 2:] (200,150)
\item[Radius 1:] 8
\item[Radius 2:] 8
\item[Modular Labelling:] (0,1,2,3,4,5,6)
\item[Aggregation Function:] Shown right.
\item[Figure~\ref{discreteDisc882636} (Left)]
\end{description}
\end{minipage}
\hfill 
\begin{minipage}[t]{0.5\linewidth}
\vspace{0pt}
\begin{tabular}{|l|l|l|l|l|l|l|l|}
  \hline
  $\oplus$ & 0 & 1 & 2 & 3 & 4 & 5 & 6\\ \hline
  0        & a & a & a & a & a & a & a\\ \hline
  1        & a & a & b & b & b & a & a\\ \hline
  2        & a & b & b & a & b & b & a\\ \hline
  3        & a & b & a & a & a & b & a\\ \hline
  4        & a & b & b & a & b & b & a\\ \hline
  5        & a & a & b & b & b & a & a\\ \hline
  6        & a & a & a & a & a & a & a\\ \hline
\end{tabular}
\end{minipage}

\vspace{0.4cm}

This reproduction of the aggregating function is not always exact. It is possible to skew and deform the representation of the image described by altering the radii of the circles that are used to form the underlying labelling of the lattice. The following image, Figure~\ref{discreteDisc882636} (right) gives an example of such a deformation.

\begin{minipage}[t]{0.5\linewidth}
\vspace{0pt}
\begin{description}
\item[Array Size:] 300
\item[Centre 1:] (100,150)
\item[Centre 2:] (200,150)
\item[Radius 1:] 26
\item[Radius 2:] 36
\item[Modular Labelling:] (0,1,2,3,4,5,6)
\item[Aggregation Function:] Shown right.
\item[Figure~\ref{discreteDisc882636} (Right)]
\end{description}
\end{minipage}
\hfill 
\begin{minipage}[t]{0.5\linewidth}
\vspace{0pt}

\begin{tabular}{|l|l|l|l|l|l|l|l|}
  \hline
  $\oplus$ & 0 & 1 & 2 & 3 & 4 & 5 & 6\\ \hline
  0        & a & a & a & a & a & a & a\\ \hline
  1        & a & a & b & b & b & a & a\\ \hline
  2        & a & b & b & a & b & b & a\\ \hline
  3        & a & b & a & a & a & b & a\\ \hline
  4        & a & b & b & a & b & b & a\\ \hline
  5        & a & a & b & b & b & a & a\\ \hline
  6        & a & a & a & a & a & a & a\\ \hline
\end{tabular}
\end{minipage}

\vspace{0.4cm}

The two following figures, Figure~\ref{b}, demonstrates changes that occur when manipulating the number of colours, changing the aggregation function and altering the underlying tuples that generate the overall shape of the colourings.

\begin{figure}[htp]
\centering
\includegraphics[scale=0.12]{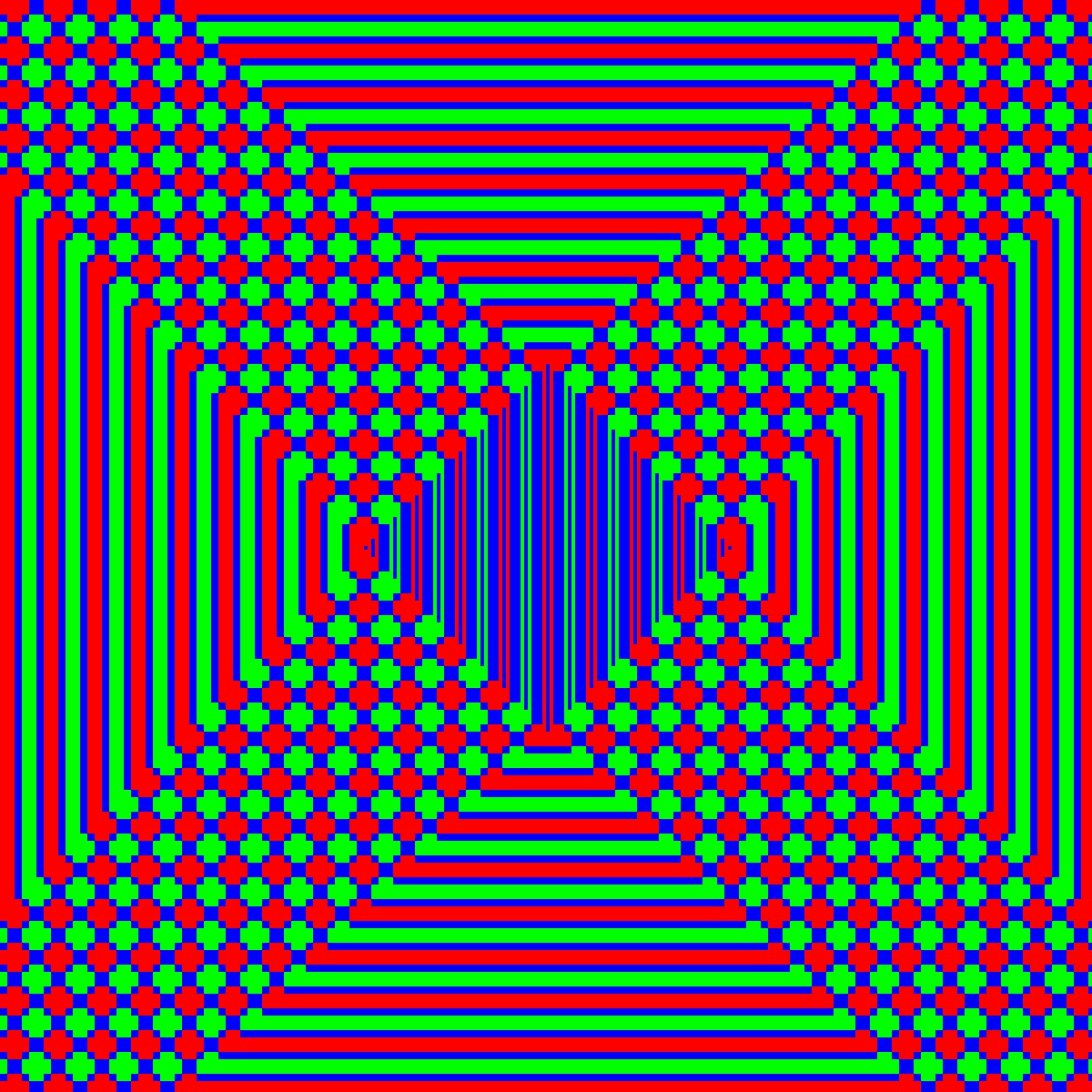}
\includegraphics[scale=0.12]{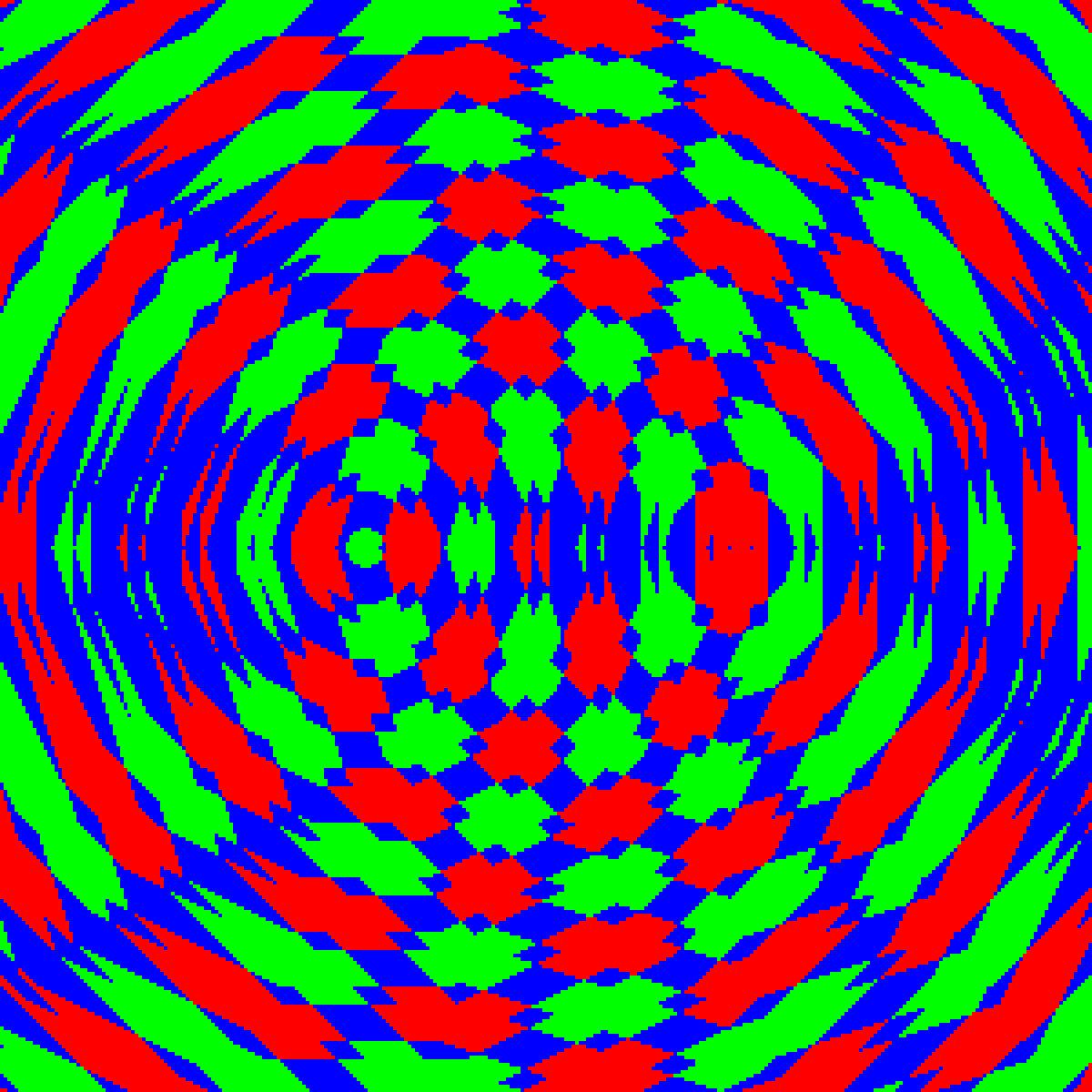}
\caption{Patterns generated by the aggregating function represented by the table and the labelling of the lattice given by two broadcasting sequences, (right) both of squared radius eight (left) one of  squared radius 26 and the other of squared radius 36.}\label{b}
\end{figure}

\begin{minipage}[t]{0.5\linewidth}
\vspace{0pt}
\begin{description}
\item[Array Size:] 300
\item[Centre 1:] (100,150)
\item[Centre 2:] (200,150)
\item[Radius 1:] 8
\item[Radius 2:] 8
\item[Modular Labelling:] (0,1,2,3,4,5)
\item[Aggregation Function:] Shown right.
\item[Figure~\ref{b} (left)]
\end{description}
\end{minipage}
\hfill 
\begin{minipage}[t]{0.5\linewidth}
\vspace{0pt}

\begin{tabular}{|l|l|l|l|l|l|l|l|}
  \hline
  $\oplus$ & 0 & 1 & 2 & 3 & 4 & 5 \\ \hline
  0        & b & b & c & c & b & b \\ \hline
  1        & b & c & a & a & c & b \\ \hline
  2        & c & a & a & a & a & c \\ \hline
  3        & c & a & a & a & a & c \\ \hline
  4        & b & c & a & a & c & b \\ \hline
  5        & b & b & c & c & b & b \\ \hline
\end{tabular}
\end{minipage}

\vspace{0.4cm}

\begin{minipage}[t]{0.5\linewidth}
\vspace{0pt}
\begin{description}
\item[Array Size:] 300
\item[Centre 1:] (100,150)
\item[Centre 2:] (200,150)
\item[Radius 1:] 26
\item[Radius 2:] 36
\item[Modular Labelling:] (0,1,2,3,4,5)
\item[Aggregation Function:] Shown right.
\item[Figure~\ref{b} (rigth)]
\end{description}
\end{minipage}
\hfill 
\begin{minipage}[t]{0.5\linewidth}
\vspace{0pt}

\begin{tabular}{|l|l|l|l|l|l|l|l|}
  \hline
  $\oplus$ & 0 & 1 & 2 & 3 & 4 & 5 \\ \hline
  0        & b & b & c & c & b & b \\ \hline
  1        & b & c & a & a & c & b \\ \hline
  2        & c & a & a & a & a & c \\ \hline
  3        & c & a & a & a & a & c \\ \hline
  4        & b & c & a & a & c & b \\ \hline
  5        & b & b & c & c & b & b \\ \hline
\end{tabular}
\end{minipage}

\vspace{0.4cm}
\begin{figure}[htp]
\centering
\includegraphics[scale=0.9]{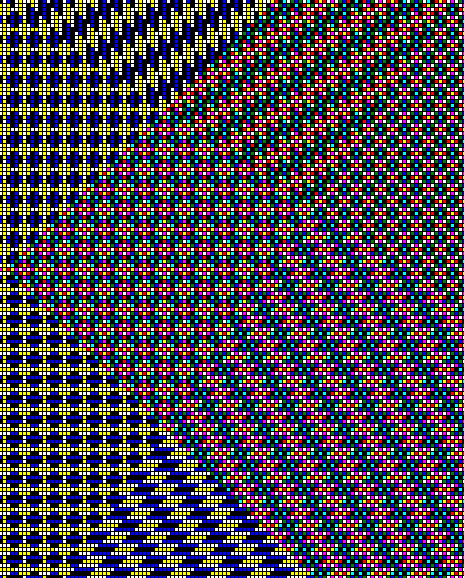}
\caption{The above figure shows a number of distinct patterns that are here generated by four different transmissions where two, one of squared radius $2$ and the other of squared radius $12$, are placed at two separate points.}\label{intEx}
\end{figure}
Altering the aggregating function is clearly a powerful tool in pattern and polygon formation. The alteration of the discs that are used as the basis of the aggregation also show that any underlying shape generate by an aggregating function can be skewed and otherwise altered whilst retaining the gestalt representation. Methods of manipulating the hew and scale of the shapes that are generated through some aggregating function may be useful to pattern recognition and detection methods that are part of the Swarm Robotics cannon among others. Whilst altering the aggregation function is one an interesting concept, for the purpose of pattern formation, there is still a lot of possibility that remains when only considering the standard moir\'{e} function. Such as can be seen in the variety of patterns that are formed in Figure~\ref{intEx}.
%


\printindex
\end{document}